\documentclass[a4paper,11pt]{article}

\usepackage{booktabs} 

\usepackage[ruled]{algorithm2e} 

\SetAlFnt{\small}
\SetAlCapFnt{\small}
\SetAlCapNameFnt{\small}
\SetAlCapHSkip{0pt}
\IncMargin{-\parindent}

\usepackage{epigraph}
\usepackage{nicefrac}
\usepackage{hyperref}
\usepackage{fullpage}
\hypersetup{colorlinks=true,linkcolor=black,citecolor=blue}

\usepackage{subfigure}
\usepackage{algpseudocode}

\usepackage{amsmath, amsthm, amssymb, amstext, comment, graphicx}
\usepackage{mathtools,extarrows}
\usepackage{url}
\newtheorem{theorem}{Theorem}
\newtheorem{definition}{Definition}
\newtheorem{open}{Open Problem}
\newtheorem{proposition}{Proposition}
\newtheorem{corollary}{Corollary}

\newtheorem{lemma}{Lemma}

\newtheorem{example}{Example}

\newtheorem{observation}{Observation}
\allowdisplaybreaks

\usepackage{url}
\allowdisplaybreaks
\usepackage{hyperref}
\hypersetup{colorlinks=true,linkcolor=black,citecolor=blue}

\title{Universal Growth in Production Economies}
\author{
	Simina Br\^anzei\footnote{Purdue University, USA. E-mail: \href{mailto:simina.branzei@gmail.com}{simina.branzei@gmail.com}.}
	\and
	Ruta Mehta\footnote{University of Illinois at Urbana-Champaign, USA. E-mail: 
		\href{mailto:rutameht@cs.illinois.edu}{rutamehta@cs.illinois.edu}.}
	\and
	Noam Nisan\footnote{Hebrew University of Jerusalem and Microsoft Research, Israel. E-mail: \href{mailto:noam@cs.huji.ac.il}{noam@cs.huji.ac.il}.}
}
\date{}

\begin{document}
	\maketitle
\begin{abstract}
We consider a simple variant of the von Neumann model of an expanding economy, in which multiple producers make goods according to their production function. The players trade their goods at the market and then use the bundles acquired for the production in the next round. 
We study a simple decentralized dynamic---known as proportional response---in which players update their bids proportionally to how useful the investments were in the past round. 

We show this dynamic leads to \emph{growth} of the economy in the long term (whenever growth is possible) but also creates unbounded \emph{inequality}, i.e. very rich and very poor players emerge.  
We analyze several other phenomena, such as how the relation of a player with others influences its development and the Gini index of the system.

One of the key technical findings is that the players learn a \emph{global} feature of the network (the optimal cycle) in a decentralized way, while interacting \emph{locally} with their direct neighbors. We obtain this by studying the volume in the resulting dynamical system and showing that the volume of each cycle expands or contracts by a constant factor in each round.
\end{abstract}

\section{Introduction}

Market equilibria are central objects of study in economic theory \cite{walras,AD54,Brainard00}. However, the notion of
an equilibrium deliberately abstracts away {\em how} such equilibria are reached or even whether 
they are reached at all.  On this topic, Fisher \cite{Fisher83} writes: \emph{``Whether or not the actual economy is stable, we largely lack a convincing theory of why that should be so. Lacking such a theory, we do not have an adequate theory of value, and there is an important lacuna in the center of microeconomic theory... To only look at situations where the Invisible Hand has finished its work cannot lead to a real understanding of how that work is accomplished.''} This issue is only compounded by the computational complexity results on the hardness of computing various notions of equilibria even in simple models \cite{CDDT09,AGT_book}.

There has, of course, been a vast amount of work done on 
understanding various  market dynamics and whether and how they converge to an equilibrium
\cite{Quah96,FZ07,CF08,zhang11,CCD13,MPP15} as well as the computational implications of such dynamics \cite{CMV05,CPV05,CF08,CCD13}. 
In this paper we are interested in {\em growth} in production markets, a scenario whose whole
point is that we do not expect (or desire) the economy to approach an equilibrium, but rather
to ``expand''.  Our model is a simple variant of the pioneering model of an expanding economy due to von Neumann \cite{vNeumann46}, where the goods are substitutes (the initial model was for perfect complements), represented through additive production functions \cite{Gale89}.  
The work of von Neumann as well as most of the literature following that focused on analyzing 
equilibrium-like states where all components of the economy expand by the same constant factor
as time progresses.  In contrast, we focus on the dynamics themselves, and analyze scenarios 
that in no sense approach any stable point.  We will be interested in understanding which market mechanisms lead to growth in the long run when the players act in a decentralized way, regardless of whether equilibrium states are ever reached or not. 

\subsection{Model Overview and Example}

In our model we have $n$ producers, each can only produce one type of (eponymous) good
according to his endowed production function, using as inputs various amounts of other goods in the 
system.  
In the simplest case, on which we focus, these production functions are simply additive.
The main decision that the ``economy
needs to make'' is that of the allocation of resources: for every time unit $t$, how do we split the 
goods produced at time $t$ between the different producers who need to use them to produce goods
for time $t+1$. The model allows for circular dependencies, which are useful for capturing situations in which people would use each other's products, such as a baker using a computer and the person making computers eating bread. Our goal is growth: having the amounts of resources in the economy  
increase with time (e.g. more bread, computers, scientific discoveries).

\medskip

Here is a simple example with two producers: 

\begin{itemize}
	\item Every time unit, producer A can make 0.1 units of good A from every unit of good B, and 
	0.99 units of good A from every unit of good A.\footnote{I.e., a unit of good A that stays with producer 
		A loses 1\% of its quantity every time unit.}
	\item Every time unit, producer B can make 10.2 units of good B from every unit
	of good A, and 0.99 units of good B from every unit of good B.  
\end{itemize}

Looking globally, it is clear that the producers 
should cooperate in this production market: if each good is allocated back to its own producer 
(i.e. agent A gets all
the existing quantity $x_{t}$ of good A, and produces from it $x_{t+1} = 0.99 x_{t}$ units of
good A, and similarly for B) then the economy will shrink.  If on the other hand we allocate the whole
quantity of good A each time unit to producer B, and all of good B to producer A, then
the economy will grow: $x_{t}$ units of good A will produce $y_{t+1}=10.2 \cdot x_t$ units of good B, 
which in turn will produce $x_{t+2} = 0.1 \cdot y_{t+1} = 1.02\cdot x_t$ units of good A.  The question
we ask in this paper is whether and how this growth can be achieved using {\em natural decentralized
	dynamics.}  After all, agent B needs significant foresight in order to send essentially everything that he produces to agent A (if he keeps even 2\% to himself then the economy will not
be able to grow).

It is not hard to globally characterize whether a general market with additive production functions can 
keep growing in an unbounded way under a globally-optimal allocation of resources --
see Theorem \ref{thm:onegood}. We would like to find a ``local mechanism'' where individual actions by the 
different producers that use only information they posses can achieve such unbounded growth, whenever growth is
possible in the first place. Leaving the space of allowed ``mechanisms'' intentionally vague, our goal is the following:


\vspace{0.1in}
\noindent
{\bf Definition:} A mechanism is called {\em universal} if for any market with additive production
functions that can grow (in an unbounded manner) under an optimal allocation of resources,
the market also grows in an unbounded manner by following the mechanism.
\vspace{0.1in}

\subsection{Our Results}

Here is the basic class of mechanisms that we study.  In these, each producer starts with a
certain amount of his good that can be used to produce other goods, as well as an amount of 
``money'' (budget) with no intrinsic value but  
used for trading.  At each time
unit, every producer splits (all of) his current money into bids that he places on each of the goods.  
Then each producer sells his good to the producers that bid on it (and collects all the money from these bids).
%

The basic decisions each producer must make
in every time unit are $(i)$ how to split his money into bids and $(ii)$ how to allocate his good to the bidders depending on their bids.
The mechanism we analyze uses the following rules: 


\begin{itemize}
	\item {\bf Trading Posts:}  In each time unit, the producers come to the market with money and goods. Each player bids on the goods and every good is allocated among its bidders 
{\em in proportion to the bid amount}. Each player collects the money from selling its good and takes home the bundle purchased, from which it produces for the next time unit.

In our example, if producer A bid \$0.2 on good A while producer B bid \$0.3 on good A, then 
	producer A gets 40\% of the quantity of good A while producer B gets 60\%; producer A collects \$0.5 from selling its good, which becomes its new budget.

	\item {\bf Proportional Updates:} In each time unit, every producer splits his money into bids on the goods in a way that is
	{\em proportional to the amount he produced from that good in the last round}. For example, if in the previous
	round producer A got 0.5 units of good A, from which it produced $0.99 \cdot 0.5$ units of good A
	and got 0.3 units of good B from which it produced $0.1 \cdot 0.3$ units of good A, then he
	will bid a fraction $(0.99 \cdot 0.5)/(0.99 \cdot 0.5 + 0.1 \cdot 0.3) \approx 94.2\%$ of his money on good A and
	the rest of $\approx 5.4\%$ on good B.
\end{itemize}

\begin{figure}[h!]
	\centering
	\subfigure[Amounts of the producers]
	{
\centering
		\includegraphics[scale = 0.46]{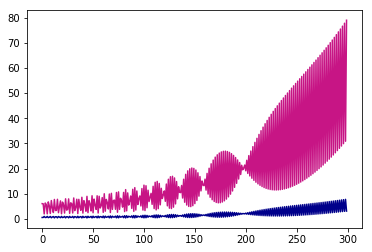}
		\label{figintro:n=2_x}
	}
	\subfigure[Budget of producer 2]
	{
		\includegraphics[scale=0.46]{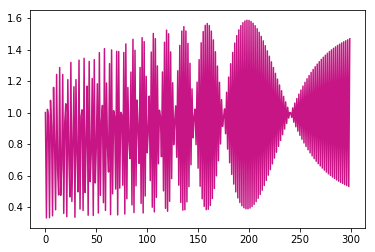}
		\label{figintro:n=2_B1}
	}
	\caption{Two producer economy given by $\vec{a} = [[0.99, 0.1], [10.2, 0.99]]$. The initial amounts are $(1,2)$, while all the initial bids are $0.5$. Producer $1$ is shown with blue and $2$ with red. The $X$ axis shows the time unit.}
	\label{figintro:growth}
\end{figure}

\begin{figure}[h!]
	\centering
	\begin{minipage}{0.47\textwidth}
		\includegraphics[scale=0.44]{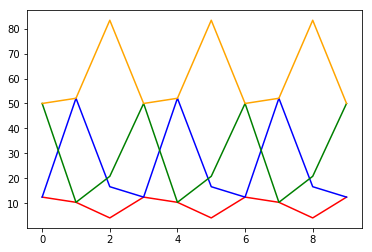}
		\caption{Cycling behavior in the bids for the two-producer economy $\vec{a} = [[1, 5], [0.2, 1]]$. The initial amounts are $(1,1)$ and budgets $(25,100)$. The bids of producer $1$ are shown in red and blue and those of producer $2$ in green and orange. The period is $3$.}\label{figintro:cycling}
	\end{minipage}
	\hfill
	\begin{minipage}{0.47\textwidth}
\centering
		\includegraphics[scale = 0.78]{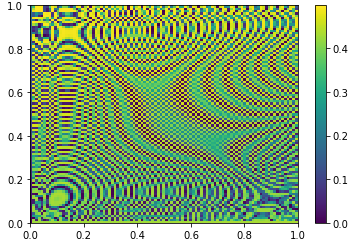}
		\caption{Gini index computed for the quantities after 800 rounds, for the two producer economy $\vec{a} = [[\sqrt{1.5}, 0.1], [15, \sqrt{1.5}]]$. The initial amounts are $(1,1)$ and initial budgets $(1,1)$. The value $x$ on the $X$ axis shows the initial bid of producer $1$ on $2$ (producer $1$ starts by bidding the remainder of $1-x$ on its own good), and $Y$ axis value the initial bid of producer $2$ on good $1$.}\label{figintro:gini}
	\end{minipage}
\end{figure}

Trading Posts have been introduced by Shapley and Shubik to explain the formation of prices in the exchange economy \cite{SS77,Shubik16}. 
The ``bid'' of a player on a good can be seen as its spending on that good.

The proportional update is also known as ``proportional response'' and has been studied before in several works. In static models such as Fisher markets, where the players repeatedly come to the market with the same amount of goods and money, Trading post with proportional updates (aka the proportional response dynamic) converges to market equilibria \cite{FZ07,zhang11}. This convergence holds for a large 
class of valuations known as constant elasticity of substitution. In fact, for additive valuations in the static Fisher model,
Trading post with proportional updates is equivalent to gradient descent (with respect to the 
Bregman divergence instead of the Euclidean distance) \cite{BDX11}.

\newpage

The combination of these two rules results in rather interesting dynamics, even when we only
have two producers as in the example above.
Figure \ref{figintro:growth} presents simulations results of this dynamic
of the two-producer example above (from some
simple starting amounts). 
Note that while the system seems to oscillate in a rather irregular fashion, one can clearly see
the quantities produced are growing over time. The budgets of the players also
oscillate in an unclear pattern. Figure \ref{figintro:cycling}
shows in greater detail the oscillations in a simulation of another two-producer market where growth
is not achieved but instead a repeated oscillation with period 3 (!) is.

\medskip

Our first main result is that Trading post with proportional updates is indeed a universal
mechanism: if a production economy can grow unboundedly using an optimal (possibly centralized) allocation rule, the economy also grows unboundedly by following this mechanism.\footnote{In other words, the performance of a universal mechanism is not too far from that of an optimal centralized allocation rule.}



\begin{theorem}[Universal growth]
\label{thm:intro2}
Trading posts with proportional updates is a universal mechanism.
\end{theorem}

We further study the properties of this mechanism, for example showing that its rate of growth is as fast as possible. Our proof analyzes the volumes of cycles in the resulting dynamical system and shows that on each cycle the volume expands or shrinks by a constant factor in each round, which implies the asymptotically optimal rate of growth of the economy.

\smallskip

Our results also show that the players learn a \emph{global} feature (the optimal cycle) on the network while only interacting \emph{locally} with their direct neighbors. The information flows along the edges in such a way that in the limit the dynamic learns implicitly the best production cycles.

\medskip

We also focus on the level of inequality between the different producers
that develops as the market grows. Figure \ref{figintro:gini} depicts a heat map of the Gini index (a measure of inequality) reached in
a certain two producer economy (after 800 rounds) as a function of the initial bids.
Despite the complexity of the dynamics we are able to show that
our mechanism always leads to growing inequality between the producers, in terms of the quantity
of goods that they get as time progresses. Specifically, we have the following inequality theorem:

\begin{theorem}[Inequality]
\label{thm:intro2}
The producers get differentiated over time into two classes: the ``rich'', who participate in the ``most
	efficient production cycle'' versus the ``poor'', who do not. While the inequality gaps {\em within}
	the ``max-efficiency'' class remain bounded by a constant, the inequality gap between this class
	and the rest grows to infinity.
\end{theorem}

We obtain the theorem by establishing that as time progresses producers of the ``max-efficiency'' class trade more and more among themselves, and their production is more efficient than that of any other group of producers. 

\smallskip

We find several other phenomena, such as the fact that other producers, who are not part of the ``max-efficiency'' class, can also grow. In fact, a producer can grow even if he is not part of any efficient production cycle, if he is instead 
``well-networked''.
For this we show the existence of 
{\em phase transitions} in the long term quantity of a producer depending on the quality of his connections.


%

\subsection{Open Problems}

We believe that we have only started scratching the surface of understanding
the dynamics of decentralized models of economic growth and open problems
abound.  First, there are a bunch of questions that we are not able to answer in our 
specific model.  For example, we would like to 
characterize the set of players that obtain increasing amount of goods versus those that
get vanishing amounts over time.

\smallskip

Second, we feel that we have only a very partial understanding of the class of mechanisms that
can be utilized.  E.g., to what extent are the results that we obtained peculiar to the specific rules
that we chose (trading posts and proportional bids)?  Which other types of rules will provide similar performance (or better in some sense such as less inequality)?  
Our work is part of a larger literature on learning how to bid in auctions and markets, and so a question is what mechanisms have good performance when players are learning how to participate in a market or game and adapting their behavior based on past performance.

\smallskip

Third, considering additive production functions can
obviously only be considered as a first step to understanding wider and more realistic classes of production
functions.  At a higher level yet, we believe that our whole framework is only a first step, and improved
modeling (as well as connection to existing macro-economic models) is interesting.

\subsection{Related Work}
Analysis of markets (economies) is central to economics. 
The growth model due to von Neuman \cite{vNeumann46} for production economies has been extensively studied, {\em e.g.,} see \cite{KMT56,Pas62,HM65,Lan12}, however mainly to understand the growth rate of the economy at equilibrium and under Leontief-type production functions. The classical Arrow-Debreu market model which involves both production and consumption has been studied for its equilibrium properties \cite{AD54}, and within CS for their computation and complexity \cite{DPSV,orlin2010improved,vegh2012concave,CSVY,Chen.plc,EY07,LeonFIXP}. 

The question of how the equilibrium prices are reached is analyzed under the natural price adjustment process due to Walras \cite{walras}, called Tatonnement -- increase the prices of over demanded goods and decrease for the under demanded goods. In particular, \cite{AH58,ABH59,AH60} showed that it converges to an equilibrium in markets with valuations restricted to the weak-gross-substitutes (WGS) property. The work within CS showed fast convergence of specific Tatonmment rules for WGS \cite{CF08,CMV05,CPV05} as well as a class of non-WGS markets \cite{CCD13}. All of these results focus on static markets where goods/money endowments of agents remain fixed. For such a static model, \cite{FZ07,zhang11} studies proportional response dynamics with trading post mechanism, showing fast convergence to a market equilibrium for a large class of valuations known as constant elasticity of substitution. 
On the contrary, our model is inherently dynamic, where amounts of both money and goods of an agent change based on the trade and production that happened in the previous round. 
Questions of growth exist in macroeconomics as well \cite{solow56,swan56,macro_book,Acemoglu} where the focus is on {\em technical progress}, {\em i.e.,} capital accumulation, population growth, etc., and typically uses Cobb-Douglas type production functions.


Like in our model, a recent work on trading networks due to \cite{ravi17} also studies trades along edges of an underlying graph, where firms trade contracts consisting of exchange, production and pricing. However, again the focus is on equilibrium existence, which is obtained when the underlying graph is a tree. 

There has been extensive work on understanding dynamics in games and auctions under various behavioural model of agents, such as best-response, multiplicative weight update, no regret learning (e.g., \cite{FS99,KPT09,DDK15,MPP15,PP16,RST17,DS16,MT12,HKMN11,bhawalkar2011welfare,lucier2010price,nisan2011best,BBN17,PpadP16,LST16,CD11,DK17,CDEG+14}). 
In the former the focus has been on convergence to an equilibrium, preferable Nash, and if not then (coarse) correlated equilibria, and the rate of covergence. In the latter the focus has been on either convergence points and their quality (price-of-anarchy), or dynamic mechanism such as ascending price auction to reach a solution (such as the Ausubel auction \cite{Ausubel04}). 
Also, the Trading post mechanism has been studied in other scenarios, such as rent seeking, allocating computational resources, matching markets, and fairness \cite{ILWM17,BKM17,FLZ09,Tullock80}.

Study of evolution naturally involves dynamics at the level of genes or species ({\em e.g.,} see \cite{Lotka10,Volterra28,GMM71,W78,Hofbauer98,akin,PNAS2:Chastain16062014,MPP15}). 
The Lotka-Volterra model \cite{Lotka10,Volterra28,W78} studies interdependence of animals and how they ``help'' or ``destroy'' each other based on their interactions and reproduction. There seems a high-level similarity between this model and ours, but the former may be thought of as a mechanism with a fixed splitting rule in our setting. 

\paragraph{Organization.} Section 2 formalizes our market model and class of mechanisms. The proportional response dynamic is defined in Section 3, and in Section 4 we analyze the growth under it and show that it is a universal mechanism. Section 5 analyzes inequality, while Sections 6 and 7 further characterizes its various properties. We conclude with Gini index simulations in Section 8.

\section{Model}

Let $N = \{1, \ldots, n\}$ be a set of players. Each player $i$ can make an eponymous good \footnote{In other words, each player $i$ makes good $i$; no other player can make good $i$ and player $i$ only knows how to make this type of good.} using the recipe given by his additive production function 
described by $\vec{a}_i = (a_{i,1}, \ldots, a_{i,n})$, where $a_{i,j} \geq 0$ is the amount made by player $i$ from one unit of good $j$. The player knows his own function.
A bundle of goods is a vector $\vec{y} = (y_1, \ldots, y_n)$, where $y_j \geq 0$ is the amount of good $j$. Given a bundle $\vec{y}$, player $i$ makes from it an amount $p_i(\vec{y}) = \sum_{j=1}^n a_{i,j} \cdot y_j$ of his good.

A production economy operates over time $t = 0, 1, 2, \ldots, \infty$. At every time unit $t$, each player $i$ produces an amount $x_i(t)$ of his good (from the bundle of ingredients he owns), then trades the good at the market. The bundle obtained by the player from trade at time $t$ is the input to his production at time $t+1$. We will assume 
that each player $i$ starts at $t=0$ by having an initial amount $x_i(0) > 0$ of his good that he directly enters trade with. 

\medskip

We will assume the economy (i.e. digraph $\vec{a}$) is \emph{strongly connected}.

\subsection{Mechanisms}

For the activity in the economy to be completely specified, we must state what mechanism is used for trade.

\begin{definition}[Abstract Mechanism]
	A market mechanism specifies how trade takes place in each round and is defined as an infinite sequence of ``splitting rules'' $\beta(t)_{t \geq 0}$, such that $\beta_{i,j}(t) \geq 0$ is the fraction received by player $i$ from good $j$ in round $t$ and $\sum_{k=1}^n \beta_{k,j}(t) \leq 1$. 
\end{definition}

Our interest is in mechanisms where the $\beta_{i,j}(t)$'s are determined by player $i$ locally from information that
he has in time $t$, and such local mechanisms will additionally include a communication protocol specifying the communication between the players.  
Our main analysis is for a specific mechanism with natural local decisions (trading posts) and very simple natural communication
(proportional bids).


\medskip
Our measure for how well the economy is doing at any point in time $t$ will be the total amount of goods in the economy: $X(t) = \sum_{i=1}^n x_i(t)$. 

\begin{definition}[Growth and decay]
	An economy (equipped with some mechanism for trade) \emph{grows} if $\lim_{t \to \infty} X(t) = \infty$ and \emph{vanishes} (or decays) if $\lim_{t \to \infty} X(t) = 0$. Similarly, a player $i$ grows if $\lim_{t \to \infty} x_i(t) = \infty$ and vanishes if $\lim_{t \to \infty} x_i(t) =0$, respectively.
\end{definition}

We are also interested in measuring the inequality in the economy, captured by the Gini index.

\begin{definition}[Gini index]
	Given a vector $\vec{u} = (u_1, \ldots, u_n)$, its Gini index is:
	$$
	G(\vec{u}) = \frac{\sum_{i=1}^n \sum_{j=1}^n |u_i - u_j|}{2n \cdot \sum_{i=1}^n u_i}
	$$
\end{definition}
The Gini index is normalized between $0$ and $1$ so that $0$ means perfect equality and $1$ maximum inequality. The latter is achieved when one entry is equal to $1$ and all other entries are zero.

\begin{example}
	Let $\vec{a} = [[1.1, 0], [0.2, 0]]$ be a two player economy with initial amounts $x_i(0) = 1$, i.e. player $1$ can make $1.1$ units (of good $1$) from one unit of good $1$ and zero units from one unit of good $2$, while player $2$ can make $0.2$ units (of good $2$) from one unit of good $1$ and zero units from one unit of good $2$. 
	If the mechanism used is to give each player $100\%$ of his own good in every round, then player $1$ will grow while player $2$ will vanish. If on the other hand the goods are split equally in every round, both players will vanish.
\end{example}

The calculations are in Appendix \ref{app:mechanisms}. The high level idea of the example is that player $1$ is ``productive'' by himself, but not so productive as to support both himself and another player receiving equal shares througout time.

\subsection{Cycles}
Important objects in our analysis will be simple cycles. In short a cycle will be a set of players $(i_1, \ldots, i_k, i_1)$, where all the $i_j$'s are different. For simplicity we will denote this by $C = (i_1, \ldots, i_k)$ and consider them as an ordered set of nodes. Given such a cycle $C$, we will override notation and denote an edge in $C$ by $(i, j) \in C$, meaning that $(1)$ $i,j \in C$ and $(2)$ $j$ is the successor of $i$ in the ordering $(i_1, \ldots, i_k, i_1)$, simply saying that we consider how useful good $i$ is for player $j$. 


\begin{definition}[Good and bad cycle]
	A cycle $C$ in the digraph $\vec{a}$ is \emph{good} if the product of weights along the cycle is strictly greater than one (i.e. $\prod_{(i,j) \in C} a_{i,j} > 1$) 
	and \emph{bad} if the product is strictly less than one.
\end{definition}
\begin{definition}[Best cycle]
	A cycle $C$ in the digraph $\vec{a}$ is the \emph{best cycle} if its geometric mean is the highest among all the cycles.
\end{definition}

\subsection{General Mechanisms}
We first observe that if all the cycles in the weighted directed graph induced by the values $a_{i,j}$ are ``good'', then every non-wasteful mechanism will lead to growth, while if all the cycles are ``bad'', no mechanism can save the economy from shrinking to zero in the limit. Note an economy is strongly connected if the directed graph with nodes $N$ and weights $\vec{a}$ is so.

\begin{proposition} \label{thm:allbad} 
	An economy vanishes with any mechanism if and only if all the cycles are bad.
\end{proposition}
A mechanism $\beta(t)_{t\geq 0}$ is \emph{non-wasteful} if it never throws away goods (i.e. $\sum_{i=1}^n \beta_{i,j}(t) = 1$ for all $j,t$) or allocates goods to players that do not need them (i.e. if $a_{i,j} = 0$ then $\beta_{i,j}(t) = 0$ for all $i,j,t$).

\begin{proposition} \label{thm:allgood}
	A strongly connected economy grows with any non-wasteful mechanism if and only if $(i)$ it has least one good cycle and $(ii)$ each directed cycle is either good or has zero edges all along.
\end{proposition}

For an economy with additive production to grow, there must exist at least one good cycle, say $C$. Moreover, there exists a mechanism that can grow such an economy by having the players along $C$ send their good to their successor in $C$. This gives the next statement.

\begin{proposition} \label{thm:onegood}
	A strongly connected economy grows with some non-wasteful mechanism if and only if it has at least one good cycle.
\end{proposition}

In our model, a universal mechanism will enable growth precisely when the economy has at least one good cycle.

\begin{definition}
	A mechanism $\mathcal{M}$ is called {\em universal} if for any economy with additive production that has at least one good cycle, the economy grows by using the splitting rule given by $\mathcal{M}$.
\end{definition}

\section{Trading Post}

From now on we focus on studying the Trading post mechanism. Each player $i$ will be initially endowed with some amount of artificial currency it can use to acquire goods that they like. In a round, every player runs a contest to decide how to allocate its good. The players submit bids on the goods they are interested in buying, after which every good $j$ is allocated in fractions proportional to the bids. Thus if the bids in some round are $b_{k,\ell}$, for all $k,\ell \in N$, then player $i$ receives the following fraction of good $j$:
\begin{equation*}
f_{i,j} = \left\{
\begin{array}{ll}
\frac{b_{i,j}}{\sum_{k=1}^{n} b_{k,j}} & \mbox{if} \; b_{i,j} > 0\\
0 & \mbox{otherwise}
\end{array}
\right.
\end{equation*}
Each player collects the money made from selling his good, and this money will be his budget in the next round. 

\smallskip
We will analyze Trading post when the players update their bids in proportion to the contribution of each good in the production from last round. This corresponds to the \emph{proportional response} dynamic, which has been studied before in Fisher markets; we generalize its definition to capture exchange settings, where all the players buy and sell. Formally, we have the following dynamical system.

\begin{definition}[Proportional dynamic]
	The initial amount of each good $i$ is $x_i(0)$ and the initial bids of player $i$ are $b_{i,j}(0)$, which sum up to an initial budget of $B_i(0)$.
	At each time $t$, the following steps take place:
	\begin{itemize}
		\item	\textbf{\emph{Exchange of goods.}} Every player $i$ brings an amount $x_i(t)$ of good $i$ and money $B_i(t)$, which is split into bids $b_{i,j}(t)$. Then player $i$ receives an amount $y_{i,j}(t)$ of each good $j$, where 
		\begin{equation*}
		y_{i,j}(t) = \left\{
		\begin{array}{ll}
		\left(\frac{b_{i,j(t)}}{\sum_{k=1}^n b_{k,j}(t)}\right) \cdot x_{j}(t), & \mbox{if} \; b_{i,j}(t) > 0\\
		0, & \mbox{otherwise}
		\end{array}
		\right.
		\end{equation*}	
		\item	\textbf{\emph{Production.}}	After trade, each player $i$ produces its good from the bundle purchased, where the amount is
		$$
		x_i(t+1) = \sum_{j=1}^{n} a_{i,j} \cdot y_{i,j}(t).
		$$
		\item	\textbf{\emph{Bid update.}} Each player collects the money made from selling: $B_i(t+1) = \sum_{k=1}^n b_{k,i}(t)$, and updates his bids proportionally to the contribution of each good in production \footnote{The bid fractions are unchanged if no production took place. This will turn out to not matter however since we will study non-degenerate starting states, which will imply that throughout time every player will be able to produce a non-zero (but possibly very small) amount.}:
		$$
		b_{i,j}(t+1) = \left( \frac{a_{i,j} \cdot y_{i,j}(t)}{x_{i}(t+1)} \right) \cdot B_i(t+1) 
		$$	
	\end{itemize}
\end{definition}

For the purpose of our results, we can w.l.o.g. normalize the money in the economy so that $\sum_{i,j=1}^n b_{i,j}(0) = 1$.
We will also assume the starting configuration is non-degenerate: $(a)$ $x_i(0) > 0$ and $(b)$ the players bid on goods worth something to them, so $b_{i,j}(0) > 0 \iff a_{i,j} > 0$. 



\section{Growth}

Our first main result is that the trading post mechanism with proportional updates leads to growth of the economy (whenever growth is possible). We establish this by analyzing the volumes of on cycles and showing they expand by a constant factor in each round for every good cycle.

We additionally show that all the players on the best cycle are guaranteed to grow (if that cycle is good), and in fact their growth is within a constant factor 
of the optimal growth (that could be achieved under some optimal allocation of resources in every round, where the constant may depend on $n$).
Players with good self-loops and players on ``good enough'' cycles will also grow infinitely.

\begin{theorem}[Universal growth] \label{thm:universal}
	Trading post with proportional updates is a universal mechanism.
\end{theorem}
\begin{proof}
The high level idea is to consider the volume of cycles as follows.
	Let $\vec{x}(0)$ and $\vec{b}(0)$ be any starting quantities and bids with the property that $x_i(0) > 0$ for all $i$ and $b_{i,j}(0) > 0$ for all $i,j$ where $a_{i,j} > 0$.
	Let $C$ be any good cycle and denote by $\alpha = \prod_{(i,j)\in C} a_{i,j} > 1$ the product along $C$.
	Take the product of all the quantities and bids on $C$, given by the function
	$$
	F(t) = \prod_{(i,j) \in C} b_{i,j}(t) \cdot x_i(t), \; \mbox{for all} \; t = 0, 1, \ldots, \infty
	$$
	Since the starting state is non-degenerate, $F(0) = \prod_{(i,j) \in C} b_{i,j}(0) \cdot x_i(0) > 0$.
	Rewriting the bids at time $t+1$ as a function of the quantities and bids at time $t$, we obtain
	\begin{eqnarray*}
		F(t+1) & = & \prod_{(i,j) \in C} b_{i,j}(t+1) \cdot  x_i(t+1) \\
		& = &
		\left(\prod_{(i,j) \in C} 
		\frac{a_{i,j} \cdot y_{i,j}(t)}{x_i(t+1)} \cdot B_i(t+1) \right)
		\left( \prod_{i \in C} x_i(t+1) \right) \\ 
		& = &
		\prod_{(i,j) \in C} 
		\left( a_{i,j} \cdot y_{i,j}(t) \cdot B_i(t+1) \right) \\
		& = & 
		\prod_{(i,j) \in C}  \left( a_{i,j} \cdot
		\left(\frac{b_{i,j}(t)}{\sum_{k=1}^n b_{k,j}(t)} \right) \cdot x_j(t)  \cdot  B_i(t+1) \right)
		\\
		& = & 
		\prod_{(i,j) \in C} \left( a_{i,j} \cdot
		\left(\frac{b_{i,j}(t)}{B_j(t+1)} \right) \cdot x_j(t)  \cdot  B_i(t+1) \right)
		\\
		& = & \prod_{(i,j) \in C}\left( a_{i,j} \cdot b_{i,j}(t) \cdot x_j(t) \right) \\
		& = & \left(\prod_{(i,j) \in C} a_{i,j}\right) \left(\prod_{(i,j) \in C} b_{i,j}(t) \cdot x_j(t) \right) \\
		& = & \alpha \cdot F(t) 
	\end{eqnarray*}
	It follows by induction that $F(t) = \alpha^t \cdot F(0)$, for all $ t = 0,1,\ldots, \infty$. Since $\alpha > 1$, we have $\lim_{t \to \infty} F(t) = \infty$.
	The total amount of money is $1$, so $b_{i,j}(t) \leq 1$ for all $i,j,t$, thus $\prod_{i \in C} x_i(t) \geq F(t)$.
	Let $\ell$ denote the length of $C$. From the geometric-arithmetic mean, we have
	$$\sum_{i \in C} x_i(t) \geq \ell \cdot \sqrt[\ell]{\prod_{i \in C} x_i(t)} \geq \ell \cdot \sqrt[\ell]{F(t)}.
	$$ 
	Then $\lim_{t \to \infty} \sum_{i \in C} x_i(t) = \infty$, so for every good cycle, the sum of quantities of the players along the cycle grows infinitely.
\end{proof}

A simulation for two players can be found in Figure \ref{fig:two}, showing how the amounts, budgets, fractions invested by the players on each other, and the Gini coefficient evolve over time in an economy where the only good cycle contains both players. In Figure \ref{fig:n=2_F} it can be seen that each player spends $100\%$ of his money on buying the good of the other player.
 Figure \ref{fig:n=2_Gx} shows the Gini coefficient in terms of the amounts, which oscillates for the whole duration of the time simulated, and Figure \ref{fig:n=2_G} the Gini coefficient in terms of budgets, which eventually reaches a stable point.

\begin{figure}[h!]
	\centering
	\subfigure[Two player economy.]
	{
		\includegraphics[scale=0.49]{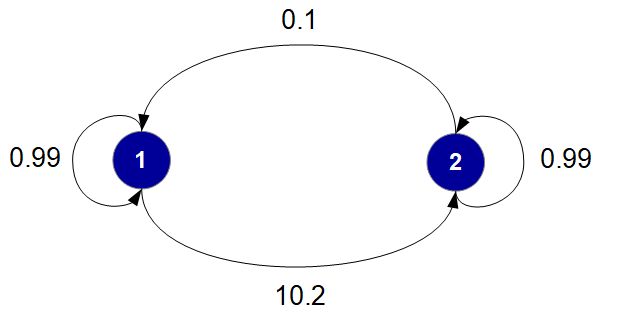}
		\label{fig:n=2}
	}
	\subfigure[Amounts over time]
	{
		\includegraphics[scale = 0.49]{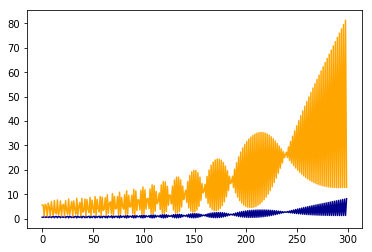}
		\label{fig:n=2_x}
	}
	\subfigure[For each player, fraction of his budget that he invests on the other player]
	{
		\includegraphics[scale=0.49]{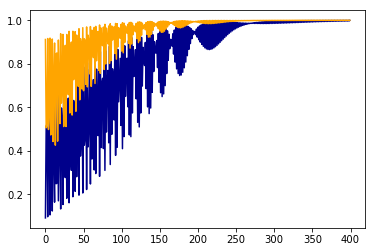}
		\label{fig:n=2_F}
	}
	\subfigure[Budget of player 1]
	{
		\includegraphics[scale=0.49]{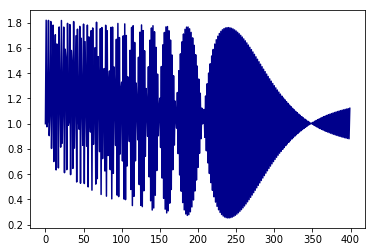}
		\label{fig:n=2_B1}
	}
	\subfigure[Budget of player 2]
	{
		\includegraphics[scale=0.49]{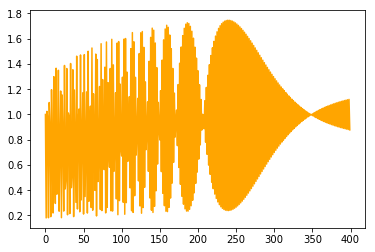}
		\label{fig:n=2_B1}
	}
	\subfigure[Gini coefficient: amounts]
	{
		\includegraphics[scale=0.49]{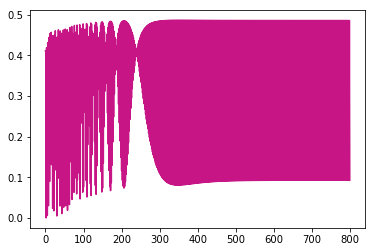}
		\label{fig:n=2_Gx}
	}
	\subfigure[Gini coefficient: budgets]
	{
		\includegraphics[scale=0.49]{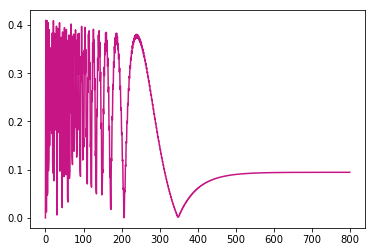}
		\label{fig:n=2_G}
	}
	\caption{Two-player economy $\vec{a} = [[0.99, 0.1], [10.2, 0.99]]$. The $X$ axis shows time (number of rounds) and $Y$ axis the variable plotted. Player 1 is marked with blue and player 2 with orange.}
\label{fig:two}
\end{figure}

Next we study the question of which players grow in the limit.

\begin{corollary} \label{cor:selfloop}
	Each player with a good self-loop grows.
\end{corollary}
\begin{proof}
	Let $i$ be any player for which $a_{i,i} > 1$. By Theorem \ref{thm:universal}, $F(t) = b_{i,i}(t) \cdot x_i(t) = a_{i,i}^t \cdot F(0)$. We have $x_i(t) \geq b_{i,i}(t) \cdot x_i(t)$ and $\lim_{t \to \infty} F(t) = \infty$, so $\lim_{t \to \infty} x_i(t) = \infty$ as required.
\end{proof}

\begin{corollary} \label{cor:all}
	In economies where all the cycles are good, every player grows.
\end{corollary}
\begin{proof}
	If all the cycles are good, each self-loop is good. By Corollary \ref{cor:selfloop}, each player grows.
\end{proof}

\begin{lemma} \label{lem:upperbound}
	Let $C$ be one of the best cycles. Then there exists a constant $\zeta > 0$ such that for each player $i \in N$ and time $t$, 
	$
	x_i(t) < \zeta \cdot \alpha^{t/|C|},
	$
	where $\alpha = \prod_{(i,j) \in C} a_{i,j}$ is the product along $C$.
\end{lemma}
\begin{proof}
	From Theorem \ref{thm:universal}, we obtain a function $F$ with the property: $F(t) = \prod_{(i,j) \in C} b_{i,j}(t) \cdot x_i(t) = \alpha^t \cdot F(0)$, where $F(0) > 0$ since the starting configuration is non-degenerate. Let $k = |C|$.
	
	The maximum growth given a number of rounds divisible by $k$ can be achieved by passing as much flow as possible through one of the best cycles, so for any player $i \in N$ and any $ t \in \mathbb{N}$, 
	$$
	x_{i}(k  t) < \left(\sum_{j=1}^n x_j\right) \cdot \alpha^{\frac{k\cdot t}{k}} =\left(\sum_{j=1}^n x_j\right) \cdot \alpha^{t}
	$$
	Then for any time $t = kt' + r$ with $0 \leq r < k$, we have the following bound, where $\Delta = \max_{(\ell_1,\ell_2)} a_{\ell_1,\ell_2}$ is the maximum edge in the graph and $\zeta = \left(\sum_{j=1}^n x_j\right) \cdot \max\{1,\Delta^{r}\}/\min\{1,\alpha^{r/k}\}$ is a constant, dependent on the graph and the starting configuration, but independent of time:
	$$
	x_i(t) < \left(\sum_{j=1}^n x_j\right) \cdot \alpha^{t'} \cdot \Delta^{r} 
	=  \left(\sum_{j=1}^n x_j\right) \cdot \alpha^{t/k} \cdot  \frac{\Delta^{r}}{\alpha^{t/k-t'}} 
	< \left(\sum_{j=1}^n x_j\right) \cdot \alpha^{t/k} \cdot \frac{\max\{1,\Delta^{r}\}}{\min\{1,\alpha^{r/k}\}} 
	= \alpha^{t/k} \cdot \zeta
	$$
	This completes the bound.
\end{proof}

\begin{lemma} \label{lem:best_bidproduct}
	Let $C$ be one of the best cycles. Then there exists a constant $\delta > 0$ such that for all $t \in \mathbb{N}$, the bids along the cycle are bounded from below by $\delta$; that is, $b_{i,j}(t) > \delta$ for all $(i,j) \in C$.
\end{lemma}
\begin{proof}
	By Lemma \ref{lem:upperbound}, there exists a constant $\tilde{\zeta} > 0$ such that for each player $i \in N$ and time $t$, the following inequality holds: $x_i(t) < \tilde{\zeta} \cdot \alpha^{t/k}$, where $\alpha = \prod_{(i,j) \in C} a_{i,j}$ and $k = |C|$ is the length of cycle $C$. Let $\zeta = \tilde{\zeta}^k$.
	Then 
	\begin{equation} \label{eq:product_bound}
	\prod_{i \in C} x_i(t) < \prod_{i \in C} \tilde{\zeta}^{k} \cdot (\alpha^{t/k})^k = \zeta \cdot \alpha^t
	\end{equation}
	
	Combining the expression for $F(t)$ with inequality (\ref{eq:product_bound}) we get:
	$$F(t) = \alpha^t \cdot F(0) = \prod_{(i,j) \in C} b_{i,j}(t) \cdot x_i(t) = \left[\prod_{(i,j) \in C} b_{i,j}(t) \right] \cdot \left[ \prod_{i \in C} x_i(t) \right] < \left[ \prod_{(i,j) \in C} b_{i,j}(t) \right] \cdot \alpha^ t \cdot \zeta.$$
	Recall that $F(0)> 0$, $b_{i,j}(t) < 1$ for all $i,j,t$. Set $\delta = F(0)/\zeta > 0$. Then the product of bids along $C$ is 
$$
	\prod_{(i,j) \in C} b_{i,j}(t) > \frac{\alpha^t \cdot F(0)}{\alpha^t \cdot \zeta} = \delta.$$
	Since $b_{i,j}(t) < 1$ for all $i,j \in N$, the bid of player $i$ on good $j$, where $(i,j) \in C$, is at least:
		$$b_{i,j}(t) > \delta/\left(\prod_{(\ell, \ell') \in C \setminus (i,j)} b_{\ell,\ell'} (t)\right) >  \delta,$$ for all $t \in \mathbb{N}$. 
	Thus the bids along $C$ are bounded from below by a constant throughout time.
\end{proof}

The amounts of the players on the best cycle will turn out to be within a constant factor of the optimal throughout time (where the optimal amounts are those that could be achieved by a central planner by ensuring trade only happens on the best cycle in each round).

\begin{proposition}[Optimal rate of growth] \label{lem:minmaxbound}
	Let $C$ be one of the best cycles and $k$ its length. Then there exist constants $\gamma,\zeta > 0$ such that for each player $i \in C$ and time $t$,
	$
	\gamma \cdot \alpha^{t/k} < x_i(t) < \zeta \cdot \alpha^{t/k}, 
	$
	where $\alpha = \prod_{(i,j) \in C} a_{i,j}$ is the product along $C$.
\end{proposition}
\begin{proof}
	Define sequences $y(t) = \max_{i \in C} x_i(t)$ and $z(t) = \min_{i \in C} x_i(t)$ containing the maximum and minimum amounts, respectively, on the cycle $C$ at each time $t$. 
	By Theorem \ref{thm:universal}, there is a function $F$ with the property that $F(t) = \prod_{(i,j) \in C} b_{i,j}(t) \cdot x_i(t) = \alpha^t \cdot F(0)$. 
	By Lemma \ref{lem:upperbound}, there is a constant $\zeta > 0$ such that for each $i \in C$ and time $t \in \mathbb{N}$,
	$
	x_i(t) < \alpha^{ \frac{t}{k} } \cdot \zeta$, so $y(t) <  \alpha^{ \frac{t}{k} } \cdot \zeta$.	
	Let $\ell_t \in C$ be the player on the cycle with minimum amount at $t$, breaking ties lexicographically, and 
	$\gamma = F(0)/\zeta^k > 0$ a constant.
	Then
	\begin{eqnarray*}
		z(t) & = & \frac{\prod_{i \in C} x_i(t)}{\prod_{i \in C \setminus \{\ell_t\}} x_i(t)} \\
		& = & \frac{\alpha^t \cdot F(0)}{\prod_{(i,j) \in C} b_{i,j}(t)} \cdot \frac{1}{\prod_{i \in C \setminus \{\ell_t\}} x_i(t)} 
		\\
		& > &  \frac{\alpha^t \cdot F(0)}{\prod_{i \in C \setminus \{\ell_t\}}  \alpha^{ \frac{t}{k} } \cdot \zeta} \\
		& = & \frac{\alpha^t \cdot F(0)}{\zeta^{k-1} \cdot \alpha^{\frac{(k-1)t}{k}}} \\
	 	& = & \alpha^{t/k} \cdot \gamma,
		%
	\end{eqnarray*}
	where $\gamma = F(0)/\zeta^{k-1}$.
	Then for any $i \in C$ and time $t$, we get that $\gamma \cdot \alpha^{t/k} < z(t) \leq x_i(t) \leq y(t) < \alpha^{t/k} \cdot \zeta$ as required.
\end{proof}

\begin{corollary}[Rate of growth of the economy]
There exists constant $c > 0$ (possibly dependent on $\vec{a}$ but independent of time), such that $$X(t) \geq c \cdot X_{OPT}(t) \; \; \mbox{for all} \; t$$ where $X_{OPT}(t)$ is the highest total amount that could be achieved by any (possibly centralized) mechanism at time $t$.
\end{corollary}

\begin{proposition}[Growth of players on the best cycles] \label{lem:all_grow}
	Suppose an economy has at least one good cycle. Then for each of the best cycles $C$, all the players in $C$ grow. 
\end{proposition}
\begin{proof}
	By Proposition \ref{lem:minmaxbound}, there exist constants $\gamma, \zeta > 0$ such that for any best cycle $C$ and any player $i \in C$, $\gamma \cdot \alpha^{t/k} < x_i(t) < \zeta \cdot \alpha^{t/k}$, where $\alpha = \prod_{(i,j) \in C} a_{i,j}$.
	Since $\alpha > 1$, it follows that $\lim_{t \to \infty} x_i(t) = \infty$ as required.
\end{proof}

The bounds on growth of Proposition \ref{lem:minmaxbound} imply that growth can be achieved also by players outside the best cycle, even if such players don't have good self loops, as long as they are part of a ``good enough'' cycle.

\begin{proposition}[Growth of players on ``good enough'' cycles]
	Suppose an economy has at least one good cycle and let $C$ be one of the best cycles.
	 Then all players on any cycle $C'$ with $
	\alpha' > \left(\alpha\right)^{\frac{|C'|-1}{|C|}}$ grow, where $\alpha = \prod_{(i,j) \in C} a_{i,j}$ and $\alpha'=\prod_{(i,j) \in C'} a_{i,j}$ are the product on $C$ and $C'$, respectively.
\end{proposition}
\begin{proof}
	Let $k = |C|$ and $\ell = |C'|$.
	By Lemma \ref{lem:upperbound}, there exists $\zeta > 0$ such that for each $i \in N$, $x_i(t) < \zeta \cdot \alpha^{t/k}$ for all $ t \in \mathbb{N}$.
	By Theorem \ref{thm:universal}, we have that $\prod_{(k,j) \in C'} b_{k,j}(t) \cdot x_k(t) = (\alpha')^t \cdot c$, for a fixed $c > 0$. Then for any player $i \in C'$, we have
	\begin{eqnarray*}
		x_i(t) & = & \frac{(\alpha')^t \cdot c}{\left(\prod_{(k,j) \in C'} b_{k,j}(t)\right) \cdot \left(\prod_{k \in C' \setminus\{i\}} x_k(t) \right)} \\
		& > & \frac{(\alpha')^t \cdot c}{\prod_{k \in C' \setminus \{i\}} \zeta \cdot \alpha^{t/k}} \\
		& = & \zeta^{1-\ell} \cdot \frac{(\alpha')^t \cdot c}{\left(\alpha^{t /k}\right)^{\ell-1}} \\
		& = & c \cdot \zeta^{1-\ell} \cdot \left(\frac{\alpha'}{\alpha^{\frac{\ell-1}{k}}}\right)^t
	\end{eqnarray*}
	From the condition that $\alpha' > \alpha^{\frac{|C'|-1}{|C|}}$, we get that $\lim_{t \to \infty} x_i(t) = \infty$ as required.
\end{proof}

We leave open the question of understanding more precisely which players grow and, in particular, whether all the players situated on a good cycle are guaranteed to grow.

\begin{open}
	Are all the players on a good cycle guaranteed to grow?
\end{open}

\section{Inequality} \label{sec:inequality}

In this section we assume the best cycle is unique. We show that the players will get split into two classes: the ``rich'' (who will turn out to be the players on the best cycle) and the ``poor'' (everyone else), such that the inequality between these classes will diverge to infinity. The ``poor'' players will be poor when compared to rich, but some (or even all) of them may grow too, just at a slower rate. 

En route to proving the inequality theorem we establish several other statements: $(1)$ in the limit the players on the best cycle will bid $100\%$ of their budget on their predecessor on the cycle and $(2)$ in the limit there is no flow of money between the players on the best cycle and rest. We conjecture that in fact the best cycle absorbs all the money in the limit.

We obtain the existence of a limit vector of money (budgets) so that the players on the best cycle rotate these budgets among themselves throughout time.
For amounts we get periodicity in a normalized version of the economy, which will imply that the Gini index for amounts cycles with period $k$, where $k$ is the length of the best cycle.

\begin{lemma} \label{lem:bounded_fraction}
	Consider an economy with a unique best cycle $C$ that is run on a sequence of arbitrary splitting rules $\beta(t)$, such that $\beta_{k,\ell}(t)$ is the fraction received by player $k$ from good $\ell$ in round $t$ for each $k,\ell$. Suppose $\alpha(C) = 1$ and for all other cycles $C'\neq C$, $\alpha(C') \le (1-\epsilon)$ for some $\epsilon>0$. If there exists an edge $(j,i)\in C$, where $j$ preceeds $i$, such that player $i$ receives infinitely often less than a $\gamma$ fraction from the good of player $j$ for some $0\le \gamma <1$, then there is a \emph{subsequence of rounds} where the amount received by player $i$ from $j$ goes to zero.
	%
\end{lemma}
\begin{proof}
The product of $a_{i,j}$s on cycle $C$ is exactly one, and it is at most $(1-\epsilon)$ on any other cycle. Intuitively any amount that passes through $C$ maintains its quantity, while the ones passing through any other cycle get reduced by at least a multiplicative factor of $(1-\epsilon)$. 
Assume without loss of generality that $x_i(0)\le 1$ for all players $i$, where $x_j(t)$ is the amount (of good $j$) produced by player $j$ at time $t$. Denote by $q(t)$ the amount of good $j$ received by player $i$. 

Using the above intuition, we will show that as $t\rightarrow \infty$, either there is a subsequence of rounds $t_1,\dots,t_k, \ldots$ such that $\lim_{k \to \infty} x_j(t_k)= 0$ and thereby $\lim_{k \to \infty} q(t_k) = 0$, or the total amount in the system goes to zero. 

Define $D = \prod_{(k,\ell): a_{k,\ell}>1} a_{k,\ell}$ and let the total initial amount be $Q=\sum_{k=1}^n x_k(0)$. Consider a self-intersecting path $P$ of arbitrary length. It consists of a set of cycles and non-overlapping segments. The total length of non-overlapping segments can be at most $n$ and the product of their $a_{i,j}$s can be at most $D$. Therefore, the product of $a_{i,j}$s on this path is at most $D(1-\epsilon)^k$ where $k$ is the number of bad cycles other than $C$ in the path $P$. 
Thus any quantity can get a boost of (at most) $D$ at most once throughout time.

Recall that at every time $t$, the fraction of good $\ell$ that player $k$ gets is denoted by $\beta_{k,\ell}(t)$. We will say that a split at time $t$ is a {\em bad split} if $\beta_{i,j}(t) \le \gamma$. Let $t_1$ be the first round when a bad split happens. Thus only $q(t_1)\le \gamma x_j(t_1)$ remains on the good cycle after $t_1$. The remaining amount $(1-\gamma)x_j(t_1)$ goes through some bad cycles and within the next $n$ rounds gets multiplied by $(1-\epsilon)$ (or less). In addition to completion of a bad cycle the amount may travel through a segment of good edges, but the overall increase due to this is at most $D$. Counting this {\em one time} boost a priori for all of $Q$, the total amount after $(t_1+n)$ rounds can be bounded as follows:

\begin{eqnarray} \label{eq:firstboundall}
\sum_{\ell=1}^n x_{\ell}(t_1+n) & \leq & (D \cdot Q - x_j(t_1)) + x_j(t_1) (\gamma + (1-\epsilon)(1-\gamma)) \nonumber \\
& = & D \cdot Q - \epsilon (1-\gamma) x_j(t_1)\nonumber
\end{eqnarray}

Given that all the cycles are less than equal to one and we already factored in any temporary boost that the total amount could receive, the bound at time $n+t_1$ on the total amount continues to hold for all times $t' \geq n+t_1$, that is
\begin{equation} \label{eq:firstboundall_onwards}
\sum_{\ell =1}^n x_{\ell}(t') \leq D \cdot Q - \epsilon (1-\gamma) x_j(t_1)
\end{equation}

Now suppose the first bad split after $(n+t_1)$ rounds happen at time $t_2$. By using a similar argument to the one for inequality (\ref{eq:firstboundall}) and invoking inequality (\ref{eq:firstboundall_onwards}), we get that after $(n+t_2)$ rounds the total amount in the system can be at most 
\begin{eqnarray}
\sum_{\ell=1}^n x_{\ell}(t_2 + n) & \leq & \left(D \cdot Q- \epsilon(1- \gamma) x_j(t_1) - x_j(t_2)\right) + x_j(t_2)(\gamma + (1-\epsilon)(1-\gamma)) \nonumber  \\
& = & D \cdot Q - \epsilon(1-\gamma) \left(x_j(t_1) + x_j(t_2)\right) \nonumber
\end{eqnarray} 


Thus after every bad split that we consider, we will ignore $n$ rounds. Let us call any such bad split that we consider with the gap of $n$ rounds a {\em counted bad split}. Applying the above argument inductively, we get that if first $(k-1)$ {\em counted bad splits} have occurred at rounds $t_1,\dots,t_{(k-1)}$, and the $k$th happens in round $t_k$, then after $(t_k+n)$ rounds the total amount of all goods is bounded by:
$$\sum_{\ell=1}^n x_{\ell}(t_k + n) \leq D \cdot Q - \epsilon(1-\gamma)\left(\sum_{d=1}^k x_j(t_d)\right)$$

Since there are infinitely many bad splits, and for every {\em counted bad split} we ignore at most $n$ bad splits, the {\em counted bad splits} are also infinitely many. Suppose that these occur at the time sequence $t_1,t_2,\dots$. Then as $k\rightarrow \infty$ either $x_j(t_k) \rightarrow 0$ or the total amount goes to zero. In either case $q(t_k)$ goes to zero as required.
\end{proof}

The above lemma is in fact optimal, in the sense that picking a subsequence of rounds is important. Neither the total amount itself nor the quantity received by $i$ from $j$ need to go to zero as demonstrated by the construction in the next example.
\begin{example} \label{eg:optimal_subsequence}
	Consider an economy with five agents $N = \{1,2,3,4,5\}$. The best cycle is 
	among players $1,2, 3$, with $a_{2,1}=a_{3,2}=a_{1,3}=1$, and there is another cycle
	among players $1,4,5$, with $a_{4,1}=a_{1,5}=1$ and $a_{5,4}=1-\epsilon$ for some $0 < \epsilon <1$. Let the initial amount of each goods with agents $1,2,$ and $3$ be one, and that of $4$ and $5$ be zero(note the statement can be made for non-zero but very small amounts too, but the calculations are simpler for zero so we illustrate this scenario). Consider the following infinite sequence of splitting rules.
	
	Starting from the first round, suppose agent $1$ splits her good between agents $2$ and $4$ in $\gamma$ and $(1-\gamma)$ fraction respectively after every three rounds. In other words $\beta_{2,1}(1+3k)=\gamma$ for all $k=0,1,...,$ and for the rest of the rounds $\beta_{2,1}(t)=1$. For all other agents, there is only one successor and hence no splitting. It is easy to see that the amount of good produced by agent $1$ in time $t$, call it $x_1(t)$ is,
	\[
	\begin{array}{lcll}
	x_1(t) & = & (\gamma+(1-\gamma)(1-\epsilon))^k\ \ \ & \mbox{If $t=3k$ for an integer $k\ge 0$}\\
	& = & 1 & \mbox{Otherwise.}
	\end{array}
	\]
	Thus neither total quantity nor the amount received by agent $2$ from agent $1$ goes to zero. But, a subsequence of $x_j(t)$, namely $t=3k$ for $k=0,1,2,\dots$, goes to zero.
\end{example}

\begin{observation} \label{note:appabsorb}
	Suppose $C$ is the unique best cycle. 
	Then there are economies where some of the players outside $C$ invest $0\%$ of their budget on the goods in $C$ as $t \to \infty$. We find this phenomenon by simulating the following economy:
	$a = [[0.1, 1, 0.1, 0.1]$, $[1, 0.1, 0.1, 0.1]$, $[0.1, 0.1, 0.1, 0.3]$, $[0.1, 0.1, 0.1, 0.1]]$, with budgets and amounts initialized to $1$ and $b_{i,j}=1/3$ for all $i \neq j$. The best cycle is $C = (1,2)$. In Figure 5 
	it can be seen that in the limit player $3$ invests $100\%$ of its budget on player $4$, even though the amount of player $4$ goes to zero.
	Also note that the budgets of players $3$ and $4$ converge to zero.
\end{observation}

\begin{figure}[h!]
	\label{fig:counterexample}
	\centering
	\subfigure[Four player economy. The best cycle is $(1,2)$.]
	{
		\includegraphics[scale=0.4]{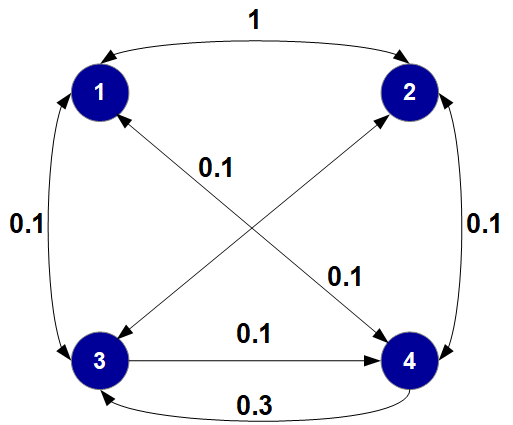}
		\label{fig:outC}
	}
	\subfigure[Budgets. Players $1$, $2$ are shown in blue.]
	{
		\includegraphics[scale=0.4]{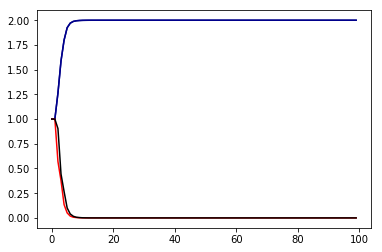}
		\label{fig:outC_budgets}
	}
	\caption{Four player economy where some players outside the best cycle (namely player $3$) invest $0\%$ of their budget on the cycle (in the limit), with budgets and amounts initialized to $1$ and $b_{i,j}=1/3$ for all $ i \neq j$.}
\end{figure}
\begin{figure}  \label{fig:counterexample_part2}
	\centering
	\subfigure[Fractions invested by $3$ on players $1,2,4$. The fraction invested on player $4$ is in blue. In the limit, player $3$ invests $100\%$ of his budget on player $4$.]
	{
		\includegraphics[scale=0.4]{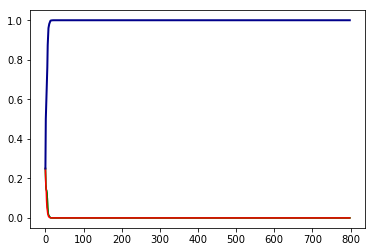}
		\label{fig:outC_f3}
	}
	\subfigure[Fraction invested by player $4$ on $1$.]
	{
		\includegraphics[scale=0.4]{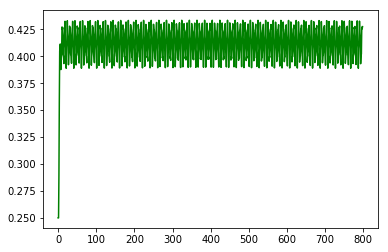}
		\label{fig:outC_f41}
	}
	\subfigure[Fraction invested by player $4$ on $2$.]
	{
		\includegraphics[scale=0.4]{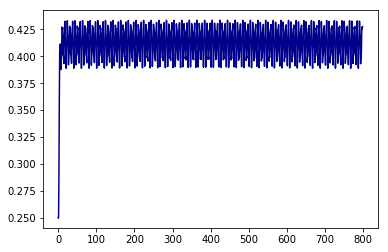}
		\label{fig:outC_f42}
	}
	\subfigure[Fraction invested by player $4$ on $3$.]
	{
		\includegraphics[scale=0.4]{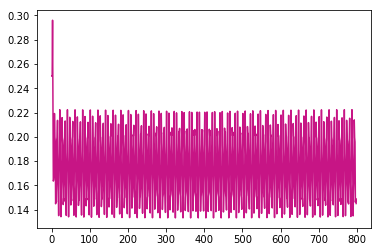}
		\label{fig:outC_f43}
	}	
\caption{In the economy from Figure \ref{fig:counterexample}, the fraction of its budget that each player invests on the goods.}
\end{figure}

\newpage

\begin{lemma} \label{lem:normal}
	Consider any economy $\vec{a}$ running trading post with proportional updates. Then there exists $w > 0$ such that by dividing each edge of $\vec{a}$ by $w$, we obtain a normalized economy $\vec{a}'$ in which any best cycle has product exactly $1$, all the other cycles have product strictly less than $1$, and at any round $t$
	\begin{itemize}
		\item the bids of the players in $\vec{a}$ and $\vec{a}'$ are identical (i.e. $b_{i,j}'(t) = b_{i,j}(t)$ for all $i,j$).
		\item the amount of any player $i$ in $\vec{a}'$ is equal to its amount in $\vec{a}$ divided by $w^t$ (i.e. $x_i'(t) = x_i(t) /w^t$).
	\end{itemize} 
\end{lemma}
\begin{proof}
	Let $a_{i,j}$ be the edges of $\vec{a}$ and $C$ any best cycle. Denote by $k$ and $w$ the length and geometric mean of $C$, respectively. Consider the economy $\vec{a}'$ with edges $a_{i,j}' = a_{i,j}/w$ for all $i,j$. Then in $\vec{a}'$ the cycle $C$ will have geometric mean $1$. For any other cycle $C'$ with length $\ell$ and product $\alpha(C') = \prod_{(i,j) \in C'} a_{i,j}$, its geometric mean in the scaled graph will be 
	$$\sqrt[\ell]{\prod_{(i,j) \in C'} \frac{a_{i,j}}{w}} = \sqrt[\ell]{\frac{\alpha(C)}{w^{\ell}}} = \frac{\sqrt[\ell]{\alpha(C)}}{w}
	$$
	Thus if the geometric mean of the cycle $C'$, $\sqrt[\ell]{\alpha(C)}$, is also $w$, then in the scaled graph $C'$ will have mean and product exactly $1$, and otherwise $C'$ will have mean and product strictly less than $1$. Let $\vec{b}(t)$ and $\vec{x}(t)$ be the initial (non-degenerate) starting state of economy $\vec{a}$. Denote by $\vec{b}'(t)$ and $\vec{x}'(t)$ the bids and amounts at time $t$ in the scaled economy $\vec{a}'$ such that $\vec{b}'(0) = \vec{b}(0)$ and $\vec{x}'(0) = \vec{x}(0)$. Since $b_{i,j}'(0) = b_{i,j}(0)$, the fractions of the goods received by each player in the first round are the same in the two economies:
	$$
	f_{i,j}'(0) = \frac{b_{i,j}'(0)}{\sum_{k=1}^{n} b_{k,j}'(0)} =  \frac{b_{i,j}(0)}{\sum_{k=1}^{n} b_{k,j}(0)} = f_{i,j}(0) 
	$$
	After one round of updates, the amounts in $\vec{a}'$ will be
	$$
	x_i'(1) = \sum_{j=1}^n f_{i,j}' \cdot a_{i,j}' \cdot x_j'(0)
	= \sum_{j=1}^n f_{i,j} \cdot \frac{a_{i,j}}{w} \cdot x_j(0) = \frac{x_i(0)}{w}
	$$
	The updated budgets are $B_i'(1) = \sum_{j=1}^n b_{j,i}'(0) = \sum_{j=1}^n b_{j,i}(0) = B_i(1)$, while the new bids are
	\begin{eqnarray*}
		b_{i,j}'(1) & = & \left( \frac{f_{i,j}' \cdot a_{i,j}' \cdot x_j'(0)}{\sum_{k=1}^n f_{i,k}' \cdot a_{i,k}' \cdot x_k'(0)} \right) \cdot B_i'(1) \\
		& = & \left( \frac{f_{i,j} \cdot \frac{a_{i,j}}{w} \cdot x_j(0)}{\sum_{k=1}^n f_{i,k} \cdot \frac{a_{i,k}}{w} \cdot x_k(0)} \right) \cdot B_i(1) \\
		& = & \left( \frac{f_{i,j} \cdot a_{i,j} \cdot x_j(0)}{\sum_{k=1}^n f_{i,k} \cdot a_{i,k} \cdot x_k(0)} \right) \cdot B_i(1) \\
		& = & b_{i,j}(1)
	\end{eqnarray*}
	With a simple inductive argument, we obtain that $b_{i,j}'(t) = b_{i,j}(t)$ and $x_i'(t) = \frac{x_i(t)}{w^t}$ for all players $i,j$ and round $t$, as required.
\end{proof}

We refer to the unique value $w$ in Lemma \ref{lem:normal} as the normalization coefficient of the economy.

\begin{proposition} \label{lem:100_percent}
	Suppose the best cycle, $C$, is unique. Then in the limit the fraction received by each player in $C$ from the good of its predecessor in $C$ converges to $100\%$.
\end{proposition}
\begin{proof}
	By Lemma \ref{lem:normal}, we can assume that the best cycle, $C$, has product one, while all the other cycles have product less than $1- \epsilon$, for some $\epsilon > 0$.
	Let $i \in C$ and denote by $j$ the predecessor of player $i$ in $C$.
	Suppose towards a contradiction that the fraction received by player $i$ from good $j$ does not converge to $1$ in the limit. Then there is a constant $\gamma < 1$ such that infinitely often, player $i$ receives less than a $\gamma$ fraction from good $j$. By Lemma \ref{lem:bounded_fraction}, it follows that there is a subsequence of rounds for which player $i$ receives in the limit of $t \to \infty$ an amount of zero from good $j$. 
	
	However by Lemma \ref{lem:best_bidproduct}, player $i$ bids at least $\delta$ on good $j$ in every round and by Theorem \ref{lem:minmaxbound}, the amount of good $j$ is in an interval bounded away from zero throughout time, that is, there exist constants $\beta > \alpha > 0$ such that $x_{j}(t) \in [\alpha,\beta]$ for all times $t$. Thus the amount $q(t)$ received by player $i$ from good $j$ at time $t$ remains bounded away from zero throughout time. We obtained a contradiction. Thus the assumption must have been false and the fraction received by player $i$ from player $j$ converges to $1$ as $t \to \infty$.
\end{proof}

\begin{proposition} \label{thm:zero_cut}
	Let $C$ be the unique best cycle. Then in the limit there is no money flowing from players in $C$ to players in $N \setminus C$ and viceversa, i.e. $\lim_{t \to \infty} b_{i,j}(t) = 0 $ and $ \lim_{t \to \infty} b_{j,i}(t) = 0$ for all $i \in C$ and $j \not \in C$.
\end{proposition}
\begin{proof}
	By Proposition \ref{lem:100_percent}, in the limit each player $i \in C$ receives $100\%$ of the good of their predecessor, say $j$, in $C$, which implies that the bids of the other players on good $j$ vanish in the limit. In particular, in the limit there is no money flowing from $N \setminus C$ to $C$. 
	
	To show that the bids of the players in $C$ on the goods in $N \setminus C$ also vanish in the limit, suppose towards a contradiction that they do not. Then there is a constant $\phi > 0$ such that some player $i \in C$ sends infinitely often a bid of at least $b_{i,j}(t) \geq \phi$ on some player $j \not \in C$. By Lemma \ref{lem:best_bidproduct}, there is a constant $\delta > 0$ such that each player in $C$ bids at least $\delta$ on their predecessor in $C$, which implies that $B_{k}(t) \geq \delta$ for all players $k \in C$ and any time $t$. Since the money flowing from $N \setminus C$ to $C$ vanishes in the limit, but player $i$ sends infinitely often at least $\phi > 0$ to player $j$, it follows that there is a subsequence in the sorted vector of budgets of players in $C$ that goes to zero, which is a contradiction. Thus the assumption was false and the money flowing from the players in $C$ to $N \setminus C$ also vanishes in the limit.
\end{proof}

In the limit, the budgets of the players on the best cycle $C$ will converge to some values that get rotated infinitely often along the cycle (from one player in $C$ to its successor in $C$).

\begin{corollary} \label{cor:budget_shift}
	Suppose $C = (1, \ldots, k)$ is the unique best cycle, where player $j$ uses the good of player $j-1$.
	There exist values $B_1^*, \ldots, B_k^*$ such that for each $r \in \{0, \ldots, k-1\}$
	$$
	\lim_{t \to \infty} B_{i}(tk+r) = B_{i-r}^*, \\ 
	\; \; \mbox{where any index} \; \; \ell \leq 0 \; \; \mbox{is interpreted as} \; \; \ell+k
	$$
\end{corollary}
\begin{proof}
	By Proposition \ref{lem:100_percent} and \ref{thm:zero_cut}, we get that in the limit each player in $C$ sends all its money to its predecessor in $C$. This means there exist values $B_1^* \ldots B_k^*$ so that in the long term the budgets of the players are described by these values (with a shift, as the money gets passed around the cycle). This implies the required limit behavior.
\end{proof}

\begin{theorem}[Inequality] \label{thm:disparity}
	Suppose an economy has a unique best cycle $C$.
	Then in the limit, the players in $C$ are arbitrarily richer than the rest. That is, for all $i \in C$ and $j \not \in C$,  
	$$\lim_{t \to \infty} \frac{x_i(t)}{x_j(t)} = \infty.$$
\end{theorem}
\begin{proof}
	Since the theorem statement requires measuring only the ratios of the amounts, by Lemma \ref{lem:normal}, we can assume w.l.o.g. that the best cycle, $C$, has product one, while all the other cycles have product less than $1- \epsilon$, for some $\epsilon > 0$.
	Let $i$ be any player in $C$ and $j$ any player in $N \setminus C$.
	By Proposition \ref{lem:minmaxbound}, there exist constants $\gamma,\zeta> 0$ such that $\gamma < x_i(t) < \zeta$ for all $t \in \mathbb{N}$.
	By Proposition \ref{lem:100_percent}, the bid of player $j$ on good $i$ vanishes in the limit. We can decompose the  good of player $j$ at time $t$ in two parts: $x_j(t) = w_j(t) + z_j(t)$, where $w_j(t)$ is the amount obtained by bidding on the goods of players in $C$ and $z_j(t)$ the amount obtained from players in $N \setminus C$. Let $s(t) = \sum_{j \in N \setminus C} x_j(t) = w(t) + z(t)$, where $w(t) = \sum_{j \in N \setminus C} w_j(t)$ and $z(t) = \sum_{j \in N \setminus C} z_j(t)$.

	From Proposition \ref{lem:100_percent} it follows that $\lim_{t \to \infty} w_j(t) = 0$ for each $j \in N\setminus C$, so $\lim_{t \to \infty} w(t) = 0$.
	Thus in the limit the players in $N \setminus C$ only receive goods from other players in $N \setminus C$. But since the subgraph induced by $N \setminus C$ has only cycles that are bad (with product less than $1 - \epsilon$), any amount that completes such a cycle is reduced by at least $\sqrt[n]{1 - \epsilon} < 1$ (multiplicatively), which implies that $\lim_{t \to \infty} z(t) = 0$ as well. Thus $\lim_{t \to \infty} x_j(t) = 0$ for each player $j \in N \setminus C$. On the other hand $x_i(t) > \gamma > 0$, so $\lim_{t \to \infty} x_i(t)/x_j(t) = \infty$ as required.
\end{proof}

In our simulation we observed that the best cycle always absorbed all the money in the limit, and if there were multiple best cycles, they absorbed all the money collectively.

\begin{open}[Concentration of money] \label{op:concentration}
	Do the best cycles absorb all the money in the limit? That is, in an economy with best cycles $C_1 \ldots C_k$, is it the case that $$\lim_{t \to \infty} \sum_{i \in C_1 \cup \ldots \cup C_k} B_i(t) = 1$$
\end{open}

A natural property that one may try to use towards settling Open Problem \ref{op:concentration} this is that the players outside the best cycle spend a non-negligible fraction (i.e. bounded from below by a constant $\gamma > 0$) of their budget on the best cycle throughout time. The simulation for the four player economy $a = [[0.1, 1, 0.1, 0.1]$, $[1, 0.1, 0.1, 0.1]$, $[0.1, 0.1, 0.1, 0.3]$, $[0.1, 0.1, 0.1, 0.1]]$ shows that this property is likely to be false; see Note \ref{note:appabsorb} in Appendix \ref{app:inequality}. However, it can be observed that the best cycle still absorbs all the money through a different mechanism: player $4$ bids a non-negligible amount of his budget on the best cycle at all times, and player $3$ bids in the limit $100\%$ of his budget on player $4$.

\smallskip

We can characterize in some sense the behavior of the amounts in the long run: in the normalized economy the players on the best cycle will ``rotate'' a vector of amounts among themselves (except the rotation is imperfect, it is multiplied by the prefix given by the position of the current round in the cycle), while the amounts of the players outside the best cycle converge to zero (the latter is immediate from the inequality theorem).

\begin{corollary} \label{cor:amount_shift}
	Let $\vec{a}$ be any economy with a unique best cycle $C = (1, \ldots, k)$, where player $j$ uses the good of player $j-1$. Let $w$ be the normalization coefficient \footnote{Such that by dividing all the edges of $\vec{a}$ by $w$, we obtain an economy $\vec{a}'$ in which the best cycle has product $1$ and all the other cycles have product less than $1$.} of $\vec{a}$. 
	There exist values $x_{1}^*, \ldots, x_{k}^*$ such that for each player $i \in C$
	\begin{enumerate}
		\item $\lim_{t \to \infty} \frac{x_{i}(tk)}{w^{tk}} = x_{i}^*$
		\item for each $r \in \{1, \ldots, k-1\}$,
		$\lim_{t \to \infty} \frac{x_{i}(tk+r)}{w^{tk+r}} = \frac{x_{i-r}^*}{w^r} \cdot \left( \prod_{j=i-r}^{i-1} a_{j+1,j} \right), \; \mbox{where any index} \; \; \ell \leq 0 \; \; \mbox{is interpreted as} \; \; \ell + k.
		$
	\end{enumerate}
	For each player $i \in N \setminus C$, $\lim_{t \to \infty} \frac{x_{i}(t)}{w^t} = 0.$
\end{corollary}
\begin{proof}
	Let $\vec{a}'$ be the normalized economy, with edges $a_{i,j}' = a_{i,j}/w$.
	By Proposition \ref{lem:100_percent} and \ref{thm:zero_cut}, we get that in the limit in the normalized economy each player in $C$ receives everything from its predecessor in $C$ and a vanishing amount from $N \setminus C$. By Corollary \ref{cor:budget_shift}, there exists a vector of budgets $B_1^* \ldots B_k^*$ that in the limit get rotated among the players in $C$. It follows that the amounts will be rotated similarly, with the caveat that the amounts will be multiplied by a prefix of the coefficients $a_{\ell,\ell-1}'$ contained between the index of the player $i$ considered and $i-r$, where $r \equiv t \; (\mod k)$ is the shift. To get property 2 required in the statement from the amounts of players in $\vec{a}$, multiply the amounts in $\vec{a}'$ by a term equal to $w^t$ in each round $t$. The limit behavior of the amounts of the players in $N \setminus C$ holds by 
	Theorem \ref{thm:disparity}.
\end{proof}
The next statement follows by Corollary \ref{cor:amount_shift}.
\begin{corollary}
	Suppose an economy has a unique best cycle $C = (1, \ldots, k)$. Then the Gini index for amounts cycles with period $k$.
\end{corollary}

We conjecture the Gini index for budgets converges as the best cycle will absorb all the money in the limit.
\begin{open}
	What happens with the Gini index for amounts when there are multiple good cycles? Does the Gini index in terms of budgets converge when the best cycle is unique?
\end{open}

\section{Phase Transitions}

In this section we show the existence of phase transitions, finding that it is possible for a player to grow even if he is not part of any good (simple) cycle, as long as he is networked ``well-enough''. Our main theorem investigating this phenomenon is for star networks, where a star is a configuration with players $1 \ldots n$, such that player $n$ is the center and trades with everyone else, while the other players only trade with player $n$. Formally, $a_{n,i}, a_{i,n} > 0$, for all $i < n$, while all other $a_{i,j} = 0$. 
\begin{theorem}[Phase transitions for stars] \label{thm:boats}
	Consider a star economy with $n$ players and at least one good cycle, such that player $n$ is the center. Let $\alpha_i = a_{i,n} \cdot a_{n,i}$ be the product on the cycle between player $i$ and player $n$ and $\alpha^* = \max_{i=1}^{n-1} \alpha_i$ the product of the best cycle. Then for any player $i = 1 \ldots n-1$, its amount
	\begin{itemize}
		\item grows if $\alpha_{i} > 1/\sqrt{\alpha^*}$.
		\item vanishes if $\alpha_{i} < 1/\sqrt{\alpha^*}$.
		\item stays in a bounded region throughout time if $\alpha_{i} = 1/\sqrt{\alpha^*}$.
	\end{itemize}
\end{theorem}

For stars we write $\lambda_i = a_{i,1}$ and $\mu_i = a_{1,i}$ for all $i = 1 \ldots n-1$. Denote by $f_i(t) = b_{n,i}(t)/B_n(t)$ the fraction invested by player $n$ on player $i \in \{1, \ldots, n-1\}$ in round $t$, with $\sum_{i=1}^{n-1} f_i(t) = 1$ for all $t$. Since the proportional update is non-wasteful after the first round, we can assume that $b_{i,n}(t) = B_i(t)$ for all $t$ and all the players $i = 1 \ldots n-1$ (even if player $i$ starts by investing in multiple goods, his bid gets corrected after the first round so that he bids zero on everything except the good of player $n$).
An example of a star with this notation is in Figure \ref{fig:star}.

\begin{figure}[h!]
	\centering
	\includegraphics[scale=0.4]{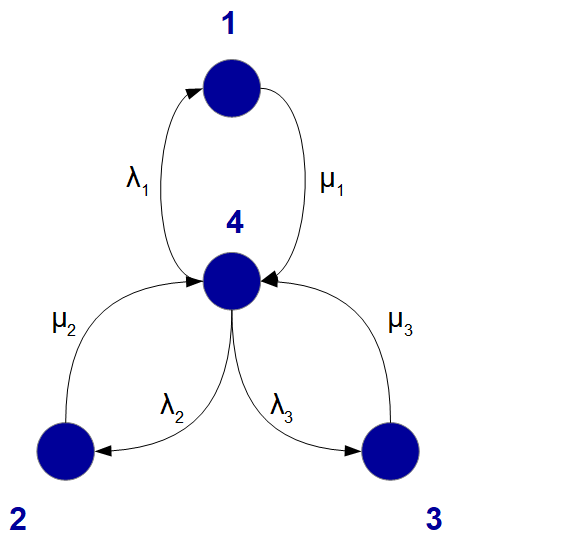}
	\caption{Star network with player $4$ at the center.}
	\label{fig:star}
\end{figure}

\begin{lemma} \label{lem:appstar_recursion}
	The fraction of player $n$'s budget invested by player $n$ on player $i = 1 \ldots n-1$ at time $t+3$, given the fractions invested by player $n$ on player $i$ at time $t$, is 
	$$
	f_i(t+3) = \frac{f_i(t) \cdot \lambda_i \mu_i}{\sum_{j=1}^{n-1} f_j(t) \cdot \lambda_j \mu_j},
	$$where $\lambda_j = a_{j,n}$ and $\mu_j = a_{n,j}$ for all $j = 1 \ldots n-1$.
\end{lemma}
\begin{proof}
	A useful observation is that player $n$ will always be able to get $100\%$ of all the goods of players $1 \ldots n-1$, since he is the only bidder competing for these goods. Thus at any time $t$, the amount of player $n$, given amounts $x_1(t-1) \ldots x_n(t-1)$ at time $t-1$, will be $x_n(t) = \sum_{i=1}^{n-1} x_i(t-1) \cdot \mu_i$.
	Now consider the amounts, budgets, and bids $x_i(t)$, $B_i(t)$, and $b_{i,j}(t)$ at any time $t$. Then the amounts at time $t+1$ are given by 
	$$
	x_n(t+1) = \sum_{j=1}^{n-1} x_j(t) \cdot \mu_j 
	$$
	and
	$$
	x_i(t+1) = \frac{B_i(t) \cdot  x_1(t) \cdot \lambda_i}{\sum_{j=1}^{n-1} B_j(t)}, \;\; \mbox{for all} \; i = 1 \ldots n-1.
	$$
	The updated fractions in player $n$'s bids following the production are:
	$$
	f_i(t+1) = \frac{x_i(t) \cdot \mu_i}{\sum_{k=1}^{n-1} x_k(t) \cdot \mu_k}, \;\; \mbox{for all} \; i = 1 \ldots n-1.
	$$
	The updated budgets are $B_n(t+1) = \sum_{j=1}^{n-1} B_j(t)$ and $B_i(t+1) = f_i(t) \cdot B_n(t)$ for $i = 1 \ldots n-1$. Recall that since the update rule is non-wasteful after the first round, we have that $b_{j,n}(t') = B_j(t')$ for each player $j$ and time $t'$, so the budgets of these players are equal to their bids on player $n$.
	
	\medskip
	
	The amount of player $n$ at time $t+2$ is
	$$
	x_n(t+2) = \sum_{j=1}^{n-1} x_j(t+1) \cdot \mu_j = \sum_{j=1}^{n-1} \left( \frac{B_j(t) \cdot x_1(t) \cdot \lambda_j}{\sum_{k=1}^{n-1} B_k(t)} \right) \mu_j = \left( \frac{x_n(t)}{\sum_{k=1}^{n-1} B_k(t)} \right) \cdot \left( \sum_{j=1}^{n-1} B_j(t) \cdot \lambda_j \mu_j \right)
	$$
	while the amounts of players $i = 1 \ldots n-1$ are
	\begin{eqnarray*}
		x_i(t+2) & = & \frac{B_i(t+1) \cdot x_1(t+1) \cdot \lambda_i}{\sum_{j=1}^{n-1} B_j(t+1)} = \frac{f_i(t) \cdot B_n(t) \cdot \left(\sum_{j=1}^{n-1} x_j(t) \cdot \mu_j\right)\lambda_i}{\sum_{j=2}^n f_j \cdot B_n(t)} \\
		& = & \frac{f_i(t) \cdot \lambda_i \cdot \left(\sum_{j=1}^{n-1} x_j(t) \cdot \mu_j\right)}{\sum_{j=1}^{n-1} f_j(t)} = f_i(t) \cdot \lambda_i \cdot 
		\left( \sum_{j=1}^{n-1} x_j(t) \cdot \mu_j
		\right)
	\end{eqnarray*}
	The updated fractions of player $n$ are
	$$
	f_i(t+2) = \frac{x_i(t+1) \cdot \mu_i}{\sum_{j=1}^{n-1} x_j(t+1)\cdot \mu_j} = 
	\frac{\left(\frac{B_i(t) \cdot  x_n(t) \cdot \lambda_i}{\sum_{j=1}^{n-1} B_j(t)} \cdot \mu_i \right)}{\sum_{j=1}^{n-1} \left( \frac{B_j(t) \cdot  x_n(t) \cdot \lambda_j \mu_j}{\sum_{k=1}^{n-1} B_k(t)}\right)} 
	= \frac{B_i(t) \cdot \lambda_i \mu_i}{\sum_{j=1}^{n-1} B_j(t) \cdot \lambda_j \mu_j}
	$$
	The budgets for round $t+2$ are  
	$$B_n(t+2) = \sum_{j=1}^{n-1} B_j(t+1) = \sum_{j=1}^{n-1} f_j(t) \cdot B_n(t) = B_n(t)$$ 
	and for each player $i =1 \ldots n-1$,
	$$B_i(t+2) = f_i(t+1) \cdot B_n(t+1) = \left( \frac{x_i(t) \cdot \mu_i}{\sum_{k=1}^{n-1} x_k(t) \cdot \mu_k} \right) \sum_{j=1}^{n-1} B_j(t).
	$$
	
	The amounts at time $t+3$ are
	$$
	x_n(t+3) = \sum_{j=1}^{n-1} x_j(t+2) \cdot \mu_j = \sum_{j=1}^{n-1} f_j(t) \cdot \lambda_j \mu_j \cdot \left( \sum_{k=1}^{n-1} x_k(t) \cdot \mu_k \right) = \left(\sum_{j=1}^{n-1} x_j(t) \cdot \mu_j \right) \left(\sum_{j=1}^{n-1} f_j(t) \cdot \lambda_j \mu_j \right),
	$$
	while for each player $i =1 \ldots n-1$,
	\begin{eqnarray*}
		x_i(t+3) & = & \frac{B_i(t+2) \cdot x_n(t+2) \cdot \lambda_i}{\sum_{j=1}^{n-1} B_j(t+2)} \\
		& = & \left( \frac{\left( \frac{x_i(t) \cdot \mu_i}{\sum_{k=1}^{n-1} x_k(t) \cdot \mu_k} \right) \left(\sum_{j=1}^{n-1} B_j(t) \right)}{\sum_{j=1}^{n-1} \left( \frac{x_j(t)\cdot \mu_j}{\sum_{k=1}^{n-1} x_k(t) \cdot \mu_k} \right) \left(\sum_{k=1}^{n-1} B_k(t)\right)} \right) \left( \frac{x_n(t)}{\sum_{k=1}^{n-1} B_k(t)}\right) \left( \sum_{j=1}^{n-1} B_j(t) \cdot \lambda_j \mu_j \right) \cdot \lambda_i \\
		& = & \frac{x_n(t) \cdot x_i(t) \cdot \lambda_i \mu_i \cdot \left( \sum_{j=1}^{n-1} B_j(t) \cdot \mu_j \lambda_j \right)}{\left(\sum_{j=1}^{n-1} B_j(t)\right) \left(\sum_{j=1}^{n-1} x_j(t) \cdot \mu_j \right)}
	\end{eqnarray*}
	
	Finally, the fractions at time $t+3$ are:
	$$
	f_i(t+3) = \frac{x_i(t+2) \cdot \mu_i}{\sum_{j=1}^{n-1} x_j(t+2) \cdot \mu_j} = \frac{ f_i(t) \cdot \lambda_i \mu_i \left(  \sum_{j=1}^{n-1} x_j(t) \cdot \mu_j \right)}{\left( \sum_{j=1}^{n-1} x_j(t) \cdot \mu_j \right) \left(\sum_{j=1}^{n-1} f_j(t) \cdot \lambda_j \mu_j \right)} = \frac{f_i(t) \cdot \lambda_i \mu_i}{\sum_{j=1}^{n-1} f_j(t) \cdot \lambda_j \mu_j},
	$$
	for each player $i = 1 \ldots n-1$, as required.
\end{proof}

\begin{lemma} \label{lem:appstar_formula}
	In a star economy with player $n$ at the center,
	the fraction invested by player $n$ on player $i$ at time $t$ is:
	\begin{equation*}
	f_i(t) =
	\frac{f_i(r) \cdot (\lambda_i \mu_i)^{\lfloor t/3 \rfloor}}{\sum_{j=1}^{n-1} f_j(r) \cdot (\lambda_j \mu_j)^{\lfloor t/3 \rfloor}}, \; \; \mbox{for} \; r = t \pmod{3}, 
	\end{equation*}
	where $\lambda_j = a_{j,n}$ and $\mu_j = a_{n,j}$ for all $j = 1 \ldots n-1$.
\end{lemma}
\begin{proof}
	Take any $r \in \{0,1,2\}$.
	We prove by induction that $f_i(3t+r) = \frac{f_i(r) \cdot (\lambda_i \mu_i)^{ t }}{\sum_{j=1}^{n-1} f_j(r) \cdot (\lambda_j \mu_j)^{ t }}$.
	The base case holds since $\sum_{j=1}^{n-1} f_j(t') = 1$ for all $t'$, so 
	$$f_i(r) = \frac{f_i(r) \cdot (\lambda_i \mu_i)^0}{\sum_{j =1}^{n-1} f_j(r) \cdot (\lambda_j \mu_j)^{0}}.$$
	
	Suppose that $f_i(3(t-1)+r) =  \frac{f_i(r) \cdot (\lambda_i \mu_i)^{ t -1}}{\sum_{j=1}^{n-1} f_j(r) \cdot (\lambda_j \mu_j)^{ t - 1}}$ holds. By Lemma \ref{lem:appstar_recursion}, we have that $f_i(t'+3) = \frac{f_i(t') \cdot \lambda_i \mu_i}{\sum_{j=1}^{n-1} f_j(t') \cdot \lambda_j \mu_j}$, from which we can derive the following expression for $f_i(r + 3t)$:
	\begin{eqnarray*}
		f_i(3t+r) & = & \frac{f_i(3(t-1)+r) \cdot \lambda_i \mu_i}{\sum_{j=1}^{n-1} f_j(3(t-1)+r) \cdot \lambda_j \mu_j}\\ 
		&= & \frac{ \left( \frac{f_i(r) \cdot (\lambda_i \mu_i)^{ t -1}}{\sum_{j=1}^{n-1} f_j(r) \cdot (\lambda_j \mu_j)^{ t - 1}} \right)\cdot \lambda_i \mu_i}{\sum_{j=1}^{n-1} \left( \frac{f_j(r) \cdot (\lambda_j \mu_j)^{ t -1}}{\sum_{k=1}^{n-1} f_k(r) \cdot (\lambda_k \mu_k)^{ t - 1}} \cdot \lambda_j \mu_j \right)} \\
		& = & 
		\frac{f_i(r) \cdot (\lambda_i \mu_i)^t}{\sum_{j=1}^{n-1} f_j(r) \cdot (\lambda_j \mu_j)^t}
	\end{eqnarray*}
	This is equivalent to the required statement, which completes the proof.
\end{proof}

\begin{proof}[Proof of Theorem \ref{thm:boats}]
	Denote by $\lambda_j = a_{j,n}$ and $\mu_j = a_{n,j}$ for all $j = 1 \ldots n-1$.
	Without loss of generality, suppose the best cycle is between player $1$ and player $n$.
	Consider any player $i \in \{2, \ldots, n-1\}$.
	Take any $r \in \{0,1,2\}$. By Lemma \ref{lem:appstar_formula}, we have that 
	$f_i(t) = \frac{f_i(r) \cdot (\lambda_i \mu_i)^{\lfloor t /3 \rfloor }}{\sum_{j=1}^{n-1} f_j(r) \cdot (\lambda_j \mu_j)^{ \lfloor t /3 \rfloor }}$, where $r = t \pmod{3}$.
	By Lemma \ref{lem:best_bidproduct}, there exists $\delta > 0$ such that $b_{1,n}(t) > \delta$ and $b_{n,1}(t) > \delta$ for all $t$, thus $B_1(t) > \delta$ and $B_n(t)> \delta$ for all $t$. Then the budget of player $i$ at time $t+1$ is 
	\begin{equation} \label{eq:budgetbound}
	B_i(t+1) = f_i(t) \cdot B_n(t) = \frac{f_i(r) \cdot (\lambda_i \mu_i)^{ \lfloor t /3 \rfloor}}{\sum_{j=1}^{n-1} f_j(r) \cdot (\lambda_j \mu_j)^{\lfloor t /3 \rfloor }} \cdot B_n(t) >  \frac{\delta \cdot f_i(r) \cdot (\lambda_i \mu_i)^{ \lfloor t/3 \rfloor }}{\sum_{j=1}^{n-1} f_j(r) \cdot (\lambda_j \mu_j)^{ \lfloor t/3 \rfloor }}
	\end{equation}
	Moreover, since the budget of player $1$ is also at least $\delta$, we additionally have $\sum_{j=1}^{n-1} B_j(t) > \delta$ for all $t$.
	By Proposition \ref{lem:minmaxbound}, there exist constants $\gamma, \zeta > 0$ such that 
	\begin{equation} \label{eq:amountbound}
	\gamma \cdot \sqrt{(\alpha^*)^t} < x_n(t) < \zeta \cdot \sqrt{(\alpha^{*})^t}
	\end{equation}
	The amount of player $i$ at time $t+1$ can be written as 
	\begin{equation} \label{eq:amount_i}
	x_i(t+1) = \frac{B_i(t)}{\sum_{j=1}^{n-1} B_j(t)} \cdot \lambda_i \cdot x_n(t) 
	\end{equation}
	
	\noindent \emph{Case 1}: $\alpha_i > 1 / \sqrt{\alpha^*}$. Using Equations \ref{eq:budgetbound}, \ref{eq:amountbound}, and \ref{eq:amount_i}, the amount of player $i$ at time $t+2$ can be bounded as follows, where $c$ is a constant independent of time, but dependent on $\lambda_j, \mu_j, \delta, \gamma$ as well as the initial amounts and bids:
	\begin{eqnarray*}
		x_i(t+2) & = & \frac{B_i(t+1)}{\sum_{j=1}^{n-1} B_j(t+1)} \cdot \lambda_i \cdot x_n(t+1) \\
		& > & B_i(t+1) \cdot \lambda_i \cdot x_n(t+1) \\
		& > & \left( \frac{\delta \cdot f_i(r) \cdot (\lambda_i \mu_i)^{ \lfloor t/3 \rfloor }}{\sum_{j=1}^{n-1} f_j(r) \cdot (\lambda_j \mu_j)^{ \lfloor t/3 \rfloor }} \right) \cdot \lambda_i \cdot x_n(t+1) \\
		& > &\left( \frac{\delta  \cdot f_i(r) \cdot \alpha_i^{ \lfloor t/3 \rfloor }}{\sum_{j=1}^{n-1} f_j(r) \cdot \alpha_j^{ \lfloor t/3 \rfloor }} \right) \cdot \lambda_i \cdot \gamma \cdot \sqrt{(\alpha^*)^{t+1}} \\
		& = & \left( \frac{\delta  \cdot f_i(r) \cdot (\alpha^*)^{ \lfloor t/3 \rfloor }}{\sum_{j=1}^{n-1} f_j(r) \cdot \alpha_j^{ \lfloor t/3 \rfloor }} \right) \cdot \lambda_i \cdot \gamma \cdot \sqrt{\alpha^*} \cdot (\alpha^*)^{t/2 - \lfloor t/3 \rfloor} \cdot \alpha_i^{\lfloor t/3 \rfloor} \\
		& > & c \cdot \left( \frac{f_i(r) \cdot (\alpha^*)^{ \lfloor t/3 \rfloor }}{\sum_{j=1}^{n-1} f_j(r) \cdot \alpha_j^{ \lfloor t/3 \rfloor }} \right) \cdot \left(\sqrt{\alpha^*} \cdot \alpha_i\right)^{\lfloor t/3 \rfloor}
	\end{eqnarray*}
	
	Since $\alpha^* > \alpha_j $ for all $j = 2\ldots n-1$, we get 
	$$\lim_{t \to \infty}   \frac{f_i(r) \cdot (\alpha^*)^{ \lfloor t/3 \rfloor }}{\sum_{j=1}^{n-1} f_j(r) \cdot \alpha_j^{ \lfloor t/3 \rfloor }} = d, \; \; \mbox{for some constant} \; \; d > 0.
	$$
	
	Moreover, $\lim_{t \to \infty} \left(\sqrt{\alpha^*} \cdot \alpha_i\right)^{\lfloor t/3 \rfloor}   = \infty$ since $\alpha_i > 1 / \sqrt{\alpha^*}$. It follows that $\lim_{t \to \infty} x_i(t+2) = \infty$, which completes the first case.
	
	\bigskip
	
	\noindent \emph{Case 2}: $\alpha_i < 1 / \sqrt{\alpha^*}$. Again there exists a constant $c' > 0$ independent of time such that we can bound the amount of player $i$ from above as follows:
	\begin{eqnarray*}
		x_i(t+2) & = & \frac{B_i(t+1)}{\sum_{j=1}^{n-1} B_j(t+1)} \cdot \lambda_i \cdot x_n(t+1) \\
		& < & \frac{B_i(t+1)}{\delta}   \cdot \lambda_i \cdot x_n(t+1) \\
		& < & \left( \frac{f_i(r) \cdot \alpha_i^{ \lfloor t /3 \rfloor}}{\sum_{j=1}^{n-1} f_j(r) \cdot \alpha_i^{\lfloor t /3 \rfloor }} \cdot \frac{B_n(t)}{\delta}  \right) \cdot \lambda_i \cdot \zeta \cdot \sqrt{(\alpha^*)^{t+1}} \\
		& < & \left( \frac{f_i(r) \cdot (\alpha^*)^{ \lfloor t /3 \rfloor}}{\sum_{j=1}^{n-1} f_j(r) \cdot \alpha_i^{\lfloor t /3 \rfloor }}  \right) \cdot \frac{\lambda_i \cdot \zeta}{\delta} \cdot (\alpha^*)^{t/2 - \lfloor t/3 \rfloor} \cdot \alpha_i^{\lfloor t/3 \rfloor} \\
		& < & c' \cdot \left( \frac{f_i(r) \cdot (\alpha^*)^{ \lfloor t /3 \rfloor}}{\sum_{j=1}^{n-1} f_j(r) \cdot \alpha_i^{\lfloor t /3 \rfloor }}  \right) \cdot \left( \sqrt{\alpha^*} \cdot \alpha_i\right)^{\lfloor t/3 \rfloor}
	\end{eqnarray*}
	Again we have that $\lim_{t \to \infty}   \frac{f_i(r) \cdot (\alpha^*)^{ \lfloor t/3 \rfloor }}{\sum_{j=1}^{n-1} f_j(r) \cdot \alpha_j^{ \lfloor t/3 \rfloor }} = d > 0$. However, $\lim_{t \to \infty} \left( \sqrt{\alpha^*} \cdot \alpha_i\right)^{\lfloor t/3 \rfloor} = 0$ since $ \sqrt{\alpha^*} \cdot \alpha_i < 1$, which implies that $\lim_{t \to \infty} x_i(t) = 0$.
	
	\bigskip
	
	\noindent \emph{Case 3}: $\alpha_i = 1 / \sqrt{\alpha^*}$. From Case 1 and 2, we obtain the following lower and upper bounds for the amount of player $i$:
	$$
	c \cdot \left( \frac{f_i(r) \cdot (\alpha^*)^{ \lfloor t/3 \rfloor }}{\sum_{j=1}^{n-1} f_j(r) \cdot \alpha_j^{ \lfloor t/3 \rfloor }} \right) \cdot \left(\sqrt{\alpha^*} \cdot \alpha_i\right)^{\lfloor t/3 \rfloor} < x_i(t+2) < c' \cdot \left( \frac{f_i(r) \cdot (\alpha^*)^{ \lfloor t /3 \rfloor}}{\sum_{j=1}^{n-1} f_j(r) \cdot \alpha_i^{\lfloor t /3 \rfloor }}  \right) \cdot \left( \sqrt{\alpha^*} \cdot \alpha_i\right)^{\lfloor t/3 \rfloor}
	$$
	Since $\alpha_i = 1 / \sqrt{\alpha^*}$, we equivalently get 
	$$
	c \cdot \left( \frac{f_i(r) \cdot (\alpha^*)^{ \lfloor t/3 \rfloor }}{\sum_{j=1}^{n-1} f_j(r) \cdot \alpha_j^{ \lfloor t/3 \rfloor }} \right) < x_i(t+2) < c' \cdot \left( \frac{f_i(r) \cdot (\alpha^*)^{ \lfloor t /3 \rfloor}}{\sum_{j=1}^{n-1} f_j(r) \cdot \alpha_i^{\lfloor t /3 \rfloor }}  \right) 
	$$
	Since $\lim_{t \to \infty}   \frac{f_i(r) \cdot (\alpha^*)^{ \lfloor t/3 \rfloor }}{\sum_{j=1}^{n-1} f_j(r) \cdot \alpha_j^{ \lfloor t/3 \rfloor }} = d \in (0,1)$, there exist constants $d', d'' > 0$ such that 
	$$d' < \frac{f_i(r) \cdot (\alpha^*)^{ \lfloor t/3 \rfloor }}{\sum_{j=1}^{n-1} f_j(r) \cdot \alpha_j^{ \lfloor t/3 \rfloor }} < d'' \; \; \mbox{for all} \; \; t.$$ 
	Then the bounds on $x_i(t+2)$ can be rewritten as
	$
	c \cdot d' < x_i(t+2) < c' \cdot d''
	$, where $c, c', d', d''$ are constants, and so player $i$'s amount stays in a constant interval bounded away from zero throughout time. This completes the proof.
\end{proof}

We conjecture that for each player with $\alpha_i = 1 / \sqrt{\alpha^*}$, the amount forms a pseudo-cycle that converges to a cycle in the limit.
\begin{figure}
	\centering
	\subfigure[Star economy.]
	{
		\includegraphics[scale=0.35]{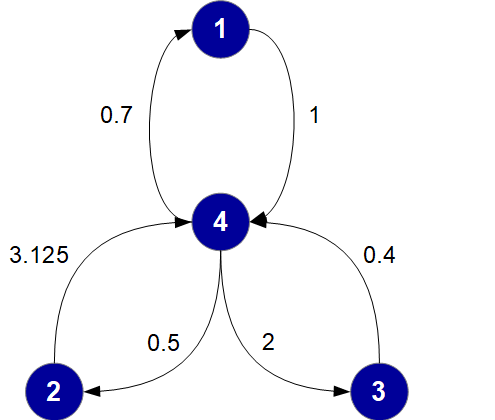}
		\label{fig:star_sim}
	}
	\subfigure[Amount of player $3$ over time (30 rounds).]
	{
		\includegraphics[scale=0.35]{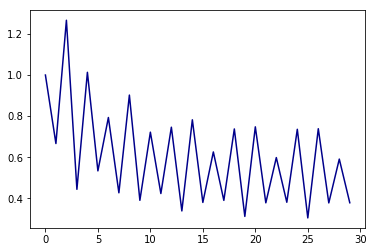}
		\label{fig:star_poorplayer_30}
	}\\
	\subfigure[Amount of player $3$ over time (300 rounds)]
	{
		\includegraphics[scale = 0.35]{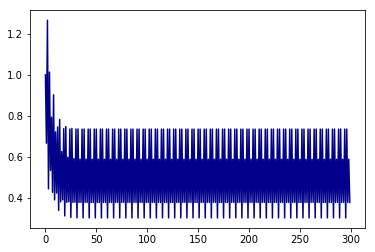}
		\label{fig:star_poorplayer_300}
	}
	\subfigure[Amount of player $4$ over time (300 rounds)]
	{
		\includegraphics[scale = 0.6]{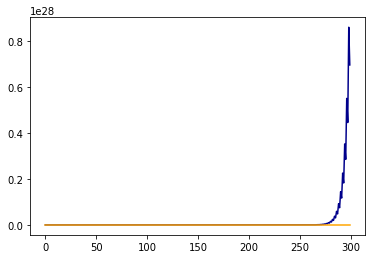}
		\label{fig:star_richplayer_300}
	}
	\label{fig:star_pseudo}
	\caption{Star with three players, with matrix $a = [[0, 0, 0.8], [0, 0, 1.5625], [1, 1, 0]]$ and initial bids $b = [[0, 0, 1], [0, 0, 1], [0.5, 0.5, 0]]$. Player $1$'s simple cycle with player $3$ (the center) has product $\alpha_1 = 0.8 \cdot 1 = 0.8$, while the best cycle is that of player $2$ with the center: $\alpha^* = 1.5625 \cdot 1 = 1.5625$, where $\alpha_1 = 1/\sqrt{\alpha^*}$. The amount of player $1$ approximately cycles.
	}
\end{figure}

\begin{figure}[h!]
	\centering
	\subfigure[Star economy where the cycle of player $3$ with $4$ is above the $0.8$ threshold.]
	{
		\includegraphics[scale=0.35]{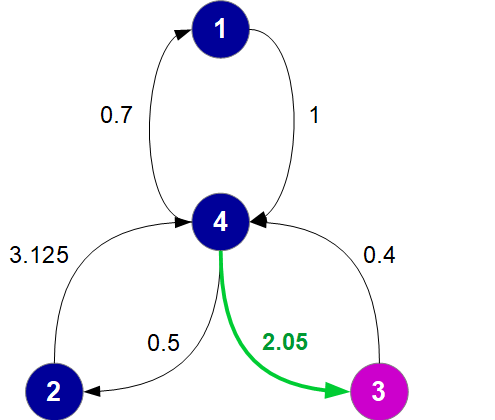}
		\label{fig:star_sim_growth_poor}
	}
	\subfigure[Amount of player $3$ over time (300 rounds).]
	{
		\includegraphics[scale=0.35]{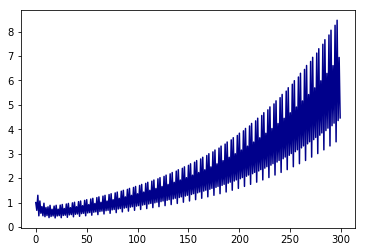}
		\label{fig:star_growth_poor}
	}\\
	\subfigure[Star economy where the cycle of player $3$ with $4$ is below the $0.8$ threshold]
	{
		\includegraphics[scale = 0.35]{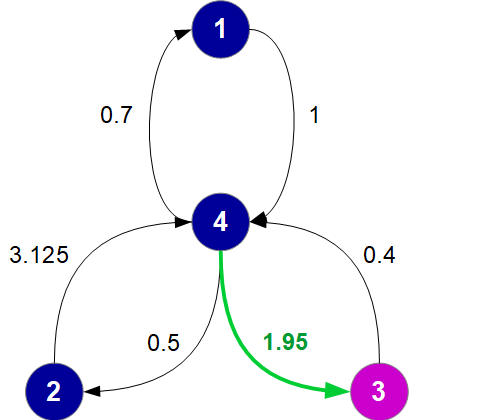}
		\label{fig:star_sim_decay_poor}
	}
	\subfigure[Amount of player $3$ over time (300 rounds).]
	{
		\includegraphics[scale = 0.35]{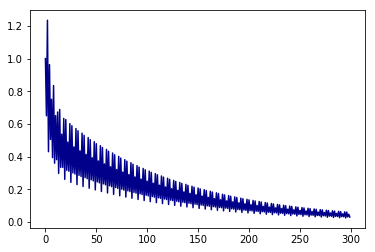}
		\label{fig:star_decay_poor}
	}
	\label{fig:xstar_varying}
	\caption{Star economy where player $3$'s threshold for growth is the value $0.8$ in its cycle with the center. Sub-figure (b) shows what happens with its quantity when the edge $(4,3)$ is changed from $2$ as it was in Figure 9 to $2.05$, while sub-figure (d) shows the evolution of its amount when this edge is changed instead to $1.95$.}
\end{figure}

The growth pattern of player $1$ when its cycle with the center is exactly at the threshold can be observed in Figure 7, and improved and degraded afterwards, in Figure 8.
We leave open the question of understanding phase transitions more generally. We include several additional examples in Appendix \ref{app:phase_sim} with phase transitions in two-player economies.

\section{Fixed Points and Cycles} \label{sec:fixed}

We briefly study fixed points and cycles. The proofs for this section are in Appendix \ref{app:fixed}.

\begin{proposition} \label{p:period}
	Let $\vec{a}$ be an economy that cycles with period $T \geq 1$. Then every cycle $C$ along which the amounts and bids are non-zero throughout time has a product one.
\end{proposition}

For fixed points in the bid space, we get the following statement for complete graphs.
\begin{proposition} \label{thm:fixedbids}
	Suppose the economy has a complete graph. If the bids are unchanging throughout time, then
	\begin{itemize}
		\item the growth of each player $i$ is given by $x_i(t+1) = a_{i,i} \cdot x_i(t)$, for all $ t \in \mathbb{N}$.
		\item for each cycle $C \subseteq N$, the coefficients satisfy the identity: $\prod_{(i,j) \in C} a_{i,j} = \prod_{i \in C} a_{i,i}$.
	\end{itemize}
\end{proposition}

Two player economies in which all cycles have product 1 always cycle as long as each player starts by bidding equally on the goods.

\begin{proposition} \label{thm:2cycle}
	Any two player economy where all cycles have product one \footnote{That is, for every $C \subseteq N$, $\prod_{(i,j) \in C} a_{i,j}  =1$.} cycles with period $3$ for any initial budgets and amounts, when the players start by splitting their budgets equally on the goods. 
\end{proposition}

Our simulation shows that in fact economies with up to five players cycle for any initial configuration of the amounts and budgets, as long as the players split their budgets equally among the goods, and the period remains $3$. 

\section{Gini Index Simulation}
We studied in simulations the Gini index of multiple economies, with additional simulations in Appendix \ref{app:gini}. 

\begin{figure}[h!]
	\centering
	\subfigure[Budgets.]
	{
		\includegraphics[scale=0.98]{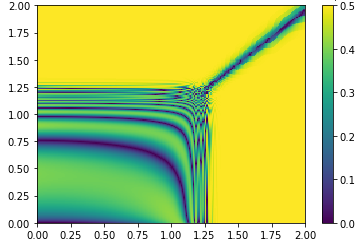}
		\label{fig:main_gini_money}
	}
	\subfigure[Amounts.]
	{
		\includegraphics[scale = 0.98]{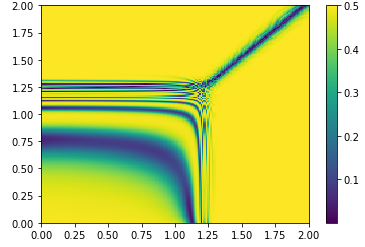}
		\label{fig:main_gini_amounts}
	}
	\caption{Gini index for $\vec{a} = [[x,0.1], [15,y]]$, where the $x$ and $y$ values are on the $X$ and $Y$ axes.}
	\label{fig:main_gini_all}
\end{figure}

Figure \ref{fig:main_gini_all} shows a two player economy 
$\vec{a} = [[x,0.1], [15,y]]$ where the $x$ and $y$ values are represented on the $X$ and $Y$ axes. All initial amounts are $1$ and all initial bids are $0.5$. The Gini coefficient is computed after 120 iterations, and lighter shades mean higher inequality. It can be observed that there is a threshold at around $1.25$ such that to the right of  this value, both on the $X$ and $Y$ axis, the inequality is very high. The reason is that the product of the cycle containing both players $1$ and $2$ is $1.5$, with a geometric mean of around $1.225$. When the self-loop of player $1$ becomes higher than this value (at the right of the figure), the inequality becomes very high because player $1$ has the best cycle and will grow at a much faster rate than player $2$ over time. The figure is symmetric, with the diagonal line showing blue because when both players have the same value for the self loops they will do equally well.

\begin{figure}[h!]
	\centering
	\subfigure[Budgets.]
	{
		\includegraphics[scale=1]{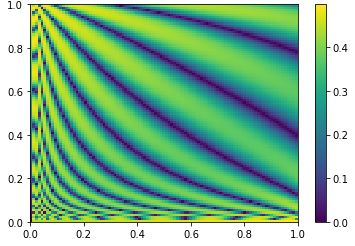}
		\label{fig:main_gini_money_bids}
	}
	\subfigure[Amounts.]
	{
		\includegraphics[scale = 1]{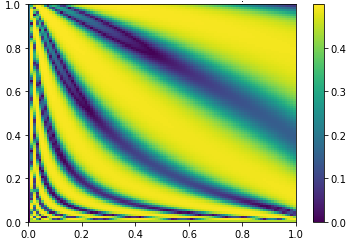}
		\label{fig:main_gini_amounts_bids}
	}
	\caption{Gini index for $\vec{a} = [[1, 0.1], [15,1]]$ with initial bids $[[1-x,x], [y, 1-y]]$.}
	\label{fig:main_gini_bids}
\end{figure}

Figure \ref{fig:main_gini_bids} shows the Gini coefficient for the two-player economy $\vec{a} = [[1, 0.1], [15,1]]$ with initial bids $[[1-x,x], [y, 1-y]]$ and initial amounts $1$. The only good cycle contains both players. The inequality in this case comes from the differences in the initial bids of the players and from the fact that player $1$ sees a ``bad'' edge (of $0.1$) from player $2$ compared to his own self-loop, while player $2$ sees a very good edge incoming from player $1$ (worth $15$). Thus player $1$ is much slower to invest on player $2$ than player $2$ is on $1$.

At a high level, 
our findings on the Gini coefficient are related to the discussion on the merits of capitalism versus socialism \cite{Marx,piketty15,hayek_serf,mises_econcalc} and their respective failures. In particular, in recent years a growing concern exists regarding the rising levels of inequality around the world, including countries such as United Kingdom \cite{ineqUK} and United States \cite{fortune16}.

\section{Acknowledgements}
We would like to thank Moshe Babaioff, Sergiu Hart, Samuel Kortum, Yuval Peres, Yuval Rabani, and Anthony Smith for useful discussions.
This project has received funding from the European Research Council (ERC) under the European Unions Horizon 2020 research and innovation programme (grant agreement No 740282), and from the ISF grant 1435/14 administered by the Israeli Academy of Sciences and Israel-USA Bi-national Science Foundation (BSF) grant 2014389.


\addcontentsline{toc}{section}{\protect\numberline{}References}%

\bibliographystyle{alpha}
\bibliography{universal_bib}
\appendix

\newpage
\section{Economy with Fixed Splitting Rule} \label{app:mechanisms}
\begin{example} \label{eg:two}
Consider a two player economy with $\vec{a} = [[1.1, 0], [0.2, 0]]$, meaning that player $1$ can make $1.1$ units (of good $1$) from one unit of good $1$ and zero units from one unit of good $2$, while player $2$ can make $0.2$ units (of good $2$) from one unit of good $1$ and zero units from one unit of good $2$. Let the initial amounts be $x_i(0) = 1$. 

If the mechanism is to give each player $100\%$ of his own good in every round, then $x_1(1) = 1.1 \cdot x_1(0) = 1.1$, $x_1(2) = 1.1 \cdot x_1(1) = 1.1^2$, and generally, $x_1(t) = 1.1 \cdot x_1(t-1) = 1.1^{t-1}$. Player $2$ receives a zero fraction of good $2$, and he cannot produce anything from his own good, so $x_2(1) = 0.2 \cdot x_1(0) + 0 \cdot x_{2}(0) = 0.2 \cdot 1 + 0 \cdot 1 = 0$, and then $x_2(t) = 0$ for all $t \geq 1$. So player $1$ grows while player $2$ vanishes.

Suppose on the other hand that the mechanism splits the goods equally each time, so in every round each player gets $50\%$ of every resource. Then the bundles obtained after trade at time $t=0$ are $\vec{y}_1(0) = (0.5, 0.5)$, $\vec{y}_2(0) = (0.5, 0.5)$, so the amounts produced at time $t=1$ will be $x_1(1) = 1.1 \cdot y_{1,1}(0) + 0 \cdot y_{1,2}(0) = 1.1 \cdot 0.5 = 0.55$ while $x_2(1) = 0.2 \cdot y_{2,1}(0) + 0 \cdot y_{2,2}(0) = 0.2 \cdot 0.5 = 0.1$. After trade at time $t=1$ the bundles obtained by players $1$ and $2$ are $\vec{y}_1(1) = (0.275,0.05)$ and $\vec{y}_2(1) = (0.275,0.05)$, which will lead to new amounts $x_1(2) = 1.1 \cdot 0.275 + 0 \cdot 0.05 = 0.3025$ and $x_2(2) = 0.2 \cdot 0.275 + 0 \cdot 0.05 = 0.055$. By induction it can be shown that $x_1(t) = 0.55^{t-1}$ and $x_2(t) = 0.2 \cdot 0.5 \cdot 0.55^{t-1}$. Thus $\lim_{t \to \infty} x_i(t) = 0$ for $i=1,2$, so both players vanish in the limit. 

The interpretation of this is that player $1$ is ``productive'' by himself, but he is not so productive as to support both himself and another player receiving equal shares throughout time.
\end{example}

\section{General Mechanisms}

In this section we show two characterizations, for economies that cannot be saved by any mechanism because all their cycles are bad, and economies in which any reasonable mechanism should achieve growth.

\begin{proposition} \label{thm:appallgood}[\ref{thm:allgood} in main text]
	A strongly connected economy grows with any non-wasteful mechanism if and only if $(i)$ it has least one good cycle and $(ii)$ each directed cycle is either good or has zero edges all along.
\end{proposition}
\begin{proof} 
	For the reverse direction, suppose all the non-zero cycles have product strictly greater than 1. An arbitrary non-wasteful mechanism can be seen as an infinite sequence of fixed splitting rules. To be precise, if $x(t)$ is the production vector at time $t$ and $\beta_{i,j}(t)$ is the fraction of good $j$ agent $i$ receives after production of $x(t)$ then, production in time $(t+1)$ is,
	\[
	x(t+1) = \langle A, \beta(t) \rangle x(t)
	\]
	where $\langle A, \beta\rangle$ is a matrix whose $(i,j)$th entry is $a_{ij} \beta_{i,j}(t)$. By unraveling this we get that total production $\sum_{i \in G} x_i(t)$ is a convex combination of product of $a_{i,j}$s along all $t$ length (non-simple) paths. If $t$ is big enough then these paths have to self-intersect and thereby contain cycles. Take any such path $P=(i_1, i_2,\dots, i_t)$. Whenever a vertex repeats, a cycle is formed. Once all the cycles are removed, what remains is a set of simple path segments that are not part of any cycle, and their total length is at most $n$ (number of nodes in the graph). The product of $a_{i,j}$s on the edges of these path segments is at least $\tau = (\min_{(i,j), a_{i,j} > 0} a_{i,j})^n$. Let $\gamma$ be the lower bound $(\Pi_{(i,j)\in C} a_{i,j})$ for any cycle $C$. Then by hypothesis $\gamma>1$. Now for any integer $h>0$, there exists an integer $k$ such that $\gamma^k \tau > h$, and therefore, for $t=n * k + n$, 
	\[
	\Pi_{q=1}^{(t-1)} a_{i_q, i_{(q+1)}} \ge \gamma^k \tau > h
	\]
	
	Since total production is a convex combination of such product terms across all $t$ length paths, it too is at least $h$, proving the claim.


	For the forward direction, suppose the economy grows with any non-wasteful mechanism. This is clearly not possible if product of $a_{i,j}$'s along every cycle is zero. Therefore, there is at least one cycle $C$ such that $\Pi_{(i,j) \in C} a_{i,j} >0$. However, suppose $\Pi_{(i,j) \in C} a_{i,j} \le 1$. We will construct a non-wasteful splitting rule under which the economy can not grow infinitely. Every node in the cycle will send all of its produced good to its successor on the cycle. To decide for the rest, shrink the cycle into one node $v_C$, and all the incoming and outgoing edges, except the cycle edges, of nodes on the cycle are now the incoming and outgoing edges respectively of node $v_C$. Note that the resulting graph, say $G_C$ remains strongly connected. 
	
	Construct the BFS tree in $G_C$ rooted at $v_C$ using only incoming edges at every node. This will create a path from every node in $G_C$ to node $v_C$. Each node not on $C$ sends all its production to its successor in the BFS tree. Since longest path in the BFS tree is of length at most $n$, clearly, after $n$ rounds all nodes outside $C$ will have no production. And after that the total production on $C$ will either decrease or remain constant. 
	
\end{proof}

If all the cycles are bad, then no mechanism can save the economy from shrinking.

\begin{proposition} \label{thm:appallbad} [\ref{thm:allbad} in main text]
	An economy vanishes with any mechanism if and only if all the cycles have product strictly less than $1$.
\end{proposition}
\begin{proof} 
	Let $\vec{a}$ be an economy where all the cycles have product less than $1$ and $\mathcal{M}$ some mechanism that is used on $\vec{a}$. Similarly to Theorem \ref{thm:allgood}, the mechanism $\mathcal{M}$ can be seen as an infinite sequence of fixed splitting rules. The total amount can get a temporary boost from traveling on a path of length at most $n-1$; this boost is bounded by a constant (that depends on the coefficients $a_{i,j}$ but is independent of time), while any amount of flow returns from a cycle reduced by a factor of at least $1- \epsilon$ for some fixed $\epsilon > 0$, so the the total amount in $\vec{a}$ goes to zero as $t \to \infty$.
	For the other direction, suppose $\vec{a}$ vanishes regardless of the mechanism. If $\vec{a}$ had a cycle with product greater than or equal to $1$, then the total amount could remain bounded away from zero by having each player along the cycle directly route their good to their successor in each round. Since this is not the case, it follows that $\vec{a}$ cannot have any such cycle.
\end{proof}

\section{Fixed Points and Cycles} \label{app:fixed}

\begin{proposition} \label{p:appperiod}[\ref{p:period} in main text]
	Let $\vec{a}$ be an economy that cycles with period $T \geq 1$. Then every cycle $C$ along which the amounts and bids are non-zero throughout time has product one.
\end{proposition}
\begin{proof}
	Since the economy cycles with period $T$, we have $x_i(t+T) = x_i(t)$ and $b_{i,j}(t+T) = b_{i,j}(t)$ for all $i,j \in N$ and $t \in \mathbb{N}$.
	Let $C \subseteq N$ be any cycle. Let $\alpha = \prod_{(i,j) \in C} a_{i,j}$ and define $F(t) = \prod_{(i,j) \in C} b_{i,j}(t) \cdot x_i(t)$. By Theorem \ref{thm:universal}, $F(t) = \alpha^t \cdot F(0)$. Then $F(t+T) =\prod_{(i,j) \in C} b_{i,j}(t+T) \cdot x_i(t+T) = \prod_{(i,j) \in C} b_{i,j}(t) \cdot x_i(t) = \alpha^{t+T} \cdot F(0) = \alpha^t \cdot F(0)$, so $\alpha^{t+T} = \alpha^t$. Since $C$ was chosen to have non-zero bids and amounts throughout time, it follows that $\alpha = 1$.
\end{proof}

For fixed points in the bid space, we get the following characterization.
\begin{proposition} \label{thm:appfixedbids}[\ref{thm:appfixedbids} in main text]
Suppose the economy has a complete graph. If the bids are unchanging throughout time, then
	\begin{itemize}
		\item the growth of each player $i$ is given by $x_i(t+1) = a_{i,i} \cdot x_i(t)$, for all $ t \in \mathbb{N}$.
		\item for each cycle $C \subseteq N$, the coefficients satisfy the identity: $\prod_{(i,j) \in C} a_{i,j} = \prod_{i \in C} a_{i,i}$.
	\end{itemize}
\end{proposition}
\begin{proof}
	Let $b_{i,j}^* > 0$ be the fixed bids for all $i,j \in N$.
	From the bid update rule and the fact that $b_{i,j}(t+1) = b_{i,j}(t) = b_{i,j}^*$, we have 
	$$
	b_{i,j}(t) = b_{i,j}(t+1) = \left( \frac{a_{i,j} \cdot y_{i,j}(t)}{x_i(t+1)} \right) \cdot B_i(t+1) =  \left( \frac{a_{i,j} \cdot \frac{b_{i,j}(t)}{B_j(t+1)} \cdot x_j(t)}{x_i(t+1)} \right) \cdot B_i(t+1),
	$$ 
	which implies 
	$$a_{i,j} = \frac{x_i(t+1) \cdot B_j(t+1)}{x_j(t) \cdot B_i(t+1)}$$
	Taking $i = j$, we get that $a_{i,i} = x_i(t+1)/x_i(t)$. Then for any cycle $C$ we get that
	$$
	\prod_{(i,j) \in C} a_{i,j} = \prod_{(i,j) \in C} \frac{x_i(t+1) \cdot B_j(t+1)}{x_j(t) \cdot B_i(t+1)} = \prod_{i \in C} \frac{x_i(t+1)}{x_i(t)} = \prod_{i \in C} a_{i,i},
	$$
	which is the second required property.
\end{proof}

\begin{figure}
	\centering
	\subfigure[Two player economy.]
	{
		\includegraphics[scale=0.45]{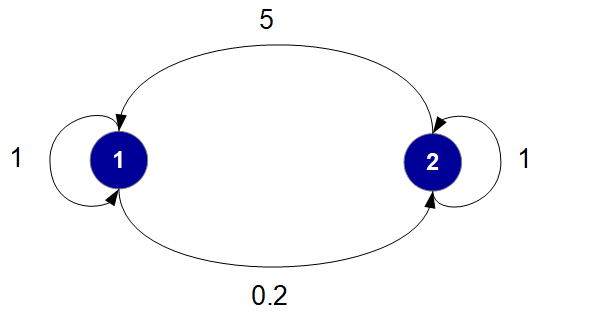}
		\label{fig:simple_cycle}
	}\\
	\subfigure[Amounts over time]
	{
		\includegraphics[scale = 0.4]{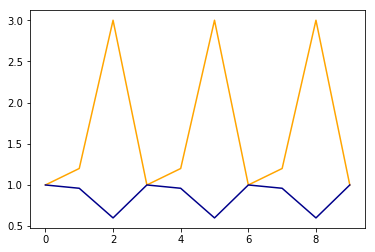}
		\label{fig:xcycle}
	}
	\subfigure[Players' bids over time, with red and blue for player $1$ and green and orange for player $2$.]
	{
		\includegraphics[scale=0.4]{bcycle.png}
		\label{fig:b1cycle}
	}
	\label{fig:cycling}
	\caption{Economy that cycles with period $3$. The initial amounts are $(1,1)$ and budgets $(25,100)$. The players start by splitting their budgets equally between the goods.}
\end{figure}

\begin{proposition} \label{thm:app2cycle}[\ref{thm:2cycle} in main text]
	Any two player economy where all cycles have product one \footnote{That is, for every $C \subseteq N$, $\prod_{(i,j) \in C} a_{i,j}  =1$.} cycles with period $3$ for any initial budgets and amounts, when the players start by splitting their budgets equally between the goods. 
\end{proposition}
\begin{proof}
	Consider an arbitrary two player economy where all the cycles have product 1. Then $a_{1,1} = a_{2,2} = 1$ and $a_{1,2} = c$, $a_{2,1} = 1/c$ for some $c > 0$. Let $B_i(0)$ and $x_i(0)$ be arbitrary initial budgets and amounts of the players. By the normalization of money, we have that $B_2(0) = 1 - B_1(0)$. The players start by investing equally on the goods, so $b_{1,1}(0) = b_{1,2}(0) = B_1(0)/2$, while $b_{2,1}(0) = b_{2,2}(0) = 0.5 - B_1(0)/2$. The amounts after the first round of production are:
	\begin{eqnarray*}
		x_1(1) &=& \frac{b_{1,1}(0)}{b_{1,1}(0) + b_{2,1}(0)} \cdot a_{1,1} \cdot x_1(0) + \frac{b_{1,2}(0)}{b_{1,2}(0) + b_{2,2}(0)} \cdot a_{1,2} \cdot x_2(0) 
		= B_1(0) (x_1(0) + c \cdot x_2(0))\\
		x_2(1) &= & (1 - B_1(0)) \left( \frac{x_1(0)}{c} + x_2(0) \right)
	\end{eqnarray*}
	The updated budgets are $B_1(1) = B_2(1) = 0.5$, while the updated bids are
	\begin{eqnarray*}
		b_{1,1}(1) & = & \frac{0.5 x_1(0)}{x_1(0) + c \cdot x_2(0)} \; \; \mbox{and} \; \; b_{1,2}(1) = \frac{0.5  \cdot c \cdot x_2(0)}{x_1(0) + c \cdot x_2(0)} \\
		b_{2,1}(1) & = & \frac{0.5  x_1(0)}{x_1(0) + c \cdot x_2(0)} \; \; \mbox{and} \; \; b_{2,2}(1)  = \frac{0.5 \cdot c \cdot x_2(0)}{x_1(0) + c \cdot x_2(0)}
	\end{eqnarray*}
	For the second round, since $b_{1,1}(1) = b_{2,1}(1)$ and $b_{1,2}(1) = b_{2,2}(1)$, we have 
	\begin{eqnarray*}
		x_1(2) & = & 0.5 \cdot x_1(1) + 0.5 \cdot c \cdot x_2(1) 
		=   0.5 B_1(0) (x_1(0) + c \cdot x_2(0)) + 0.5 c  (1 - B_1(0)) \left( \frac{x_1(0)}{c} + x_2(0) \right) \\
		&= &0.5 (x_1(0) + c \cdot x_2(0)) \\
		\bigskip
		x_2(2) & = & 0.5 \cdot \frac{1}{c} \cdot x_1(1) + 0.5 \cdot x_2(1) 
		= \frac{0.5}{c}(x_1(0) + c \cdot x_2(0))
	\end{eqnarray*}
	The updated budgets are $B_1(2) = x_1(0)/(x_1(0) + c \cdot x_2(0))$ and $B_2(2) = c \cdot x_2(0)/(x_1(0) + c \cdot x_2(0))$, while the bids become
	\begin{eqnarray*}
		b_{1,1}(2)& =& 
		= \frac{B_1(0) x_1(0)}{x_1(0) + c \cdot x_2(0)} \; \; \mbox{and} \; \;
		b_{1,2}(2) = \frac{(1-B_1(0)) x_1(0)}{x_1(0) + c \cdot x_2(0)}\\
		b_{2,1}(2) &=&
		= \frac{B_1(0) \cdot c \cdot x_2(0)}{x_1(0) + c \cdot x_2(0)} \; \; \mbox{and} \; \; 
		b_{2,2}(2) = \frac{(1-B_1(0)) \cdot c \cdot x_2(0)}{x_1(0) + c \cdot x_2(0)}
	\end{eqnarray*}
	Finally, for the third round, it can be verified that the economy returns to the initial state, with $x_i(3) = x_i(0)$ and $b_{i,j}(3) = b_{i,j}(0)$ for all $i,j$.
\end{proof}
An example of a cycling economy with two players can be found in Figure 11. 

\section{Simulations for Phase Transitions} \label{app:phase_sim}

In Figure \ref{fig:cycling2pseudo} it can be seen that the fraction it invests on player $1$ approximately cycles (the peaks are not $100\%$ identical).
\begin{figure}[h!]
	\centering
	\subfigure[200 iterations]
	{
		\includegraphics[scale=0.75]{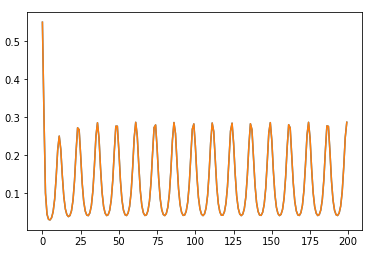}
		\label{fig:cycle_macro}
	}
	\subfigure[40 iterations]
	{
		\includegraphics[scale = 0.75]{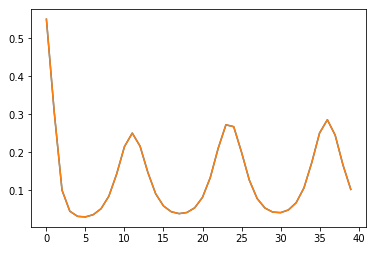}
		\label{fig:cycle_micro}
	}
	\caption{Two player economy with $a = [[1.2, 0.2], [1.1, 0.9]]$, where the fraction of the budget invested by player $1$ on good $2$ (approximatey) cycles.}
	\label{fig:cycling2pseudo}
\end{figure}

Figure 13
shows an example of how the amounts, budgets, fractions invested by the players on each other, and the Gini coefficient evolve over time in an economy where the only good cycle is the self loop of player $1$, while Figure 14
shows the same economy with the exception that the self loop of player $2$ has been reduced to $0.83$ (from $0.85$); the result is that player $2$ now vanishes.

\begin{figure}[h!]
	\centering
	\subfigure[Two player economy.]
	{
		\includegraphics[scale=0.45]{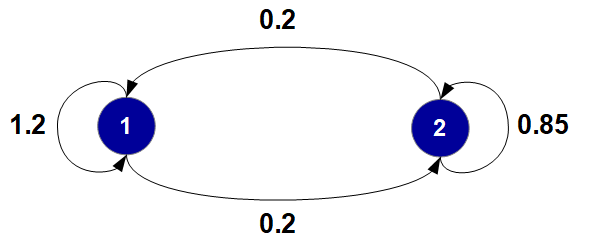}
		\label{fig:n=2selfloop}
	}
	\subfigure[Amount of player 2 over time.]
	{
		\includegraphics[scale = 0.4]{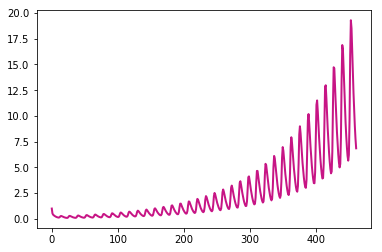}
		\label{fig:n=2selfloop_x}
	}
	\subfigure[Fraction invested by each player on the other player.]
	{
		\includegraphics[scale=0.4]{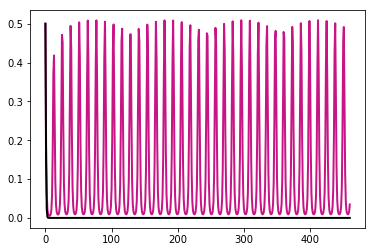}
		\label{fig:n=2selfloop_F}
	}
	\subfigure[Budgets of the players.]
	{
		\includegraphics[scale=0.4]{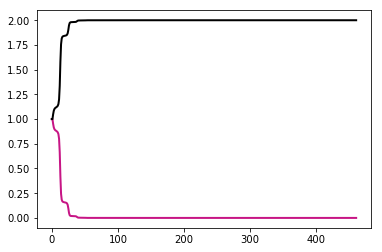}
		\label{fig:n=2_Bselfloop}
	}
	\label{fig:xtwo_selfloop}
	\caption{Two player economy with amounts, budgets, fractions invested by the players on each other, 
		and Gini coefficient over time. The only good cycle is the self loop of player 1 but player 2 also grows. Player 1 is shown in black and player 2 in purple.}
\end{figure}

\begin{figure}[h!]
	\centering
	\subfigure[Two player economy.]
	{
		\includegraphics[scale=0.45]{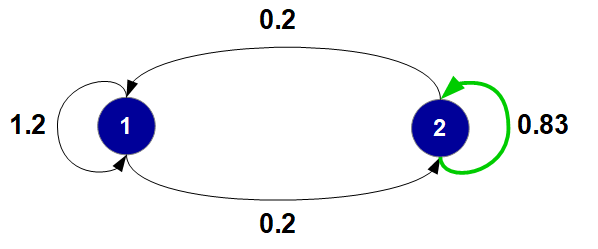}
		\label{fig:n=2selfloop_decay}
	}
	\subfigure[Amount of player 2 over time.]
	{
		\includegraphics[scale = 0.4]{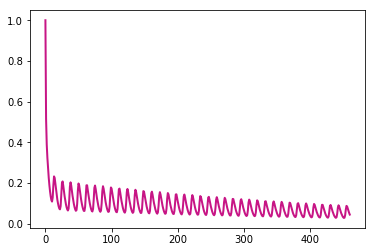}
		\label{fig:n=2selfloop_x_decay}
	}
	\subfigure[Fraction invested by each player on the other player.]
	{
		\includegraphics[scale=0.4]{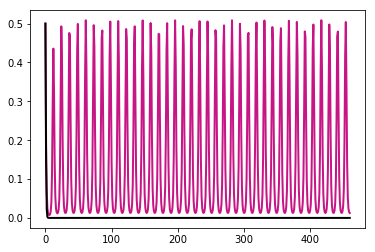}
		\label{fig:n=2selfloop_F_decay}
	}
	\label{fig:xtwo_selfloop_decay}
	\caption{Economy with the same initial state as the one in the previous Figure (13), except the self loop of player $2$ has been decreased from $0.85$ to $0.83$, and player $2$ now decays. Player $1$ is shown in black and player $2$ in purple.}
\end{figure}

\bigskip

\bigskip

\bigskip

\bigskip
\newpage

\newpage

\section{Simulations for the Gini Index} \label{app:gini}

In this section we provide simulations for the Gini coefficient of two player economies. We plot the Gini coefficient for both amounts and budgets and consider several scenarios: varying self loops (keeping the edges between the players fixed), varying edges between the players (keeping the self loops fixed), and varying initial bids (keeping the economy fixed). 

In Figures 15-17 the coefficient is computed after 120 iterations, and in the remaining ones (18-25) after 350 iterations.

\begin{figure}[h!]
	\centering
	\subfigure[Budgets.]
	{
		\includegraphics[scale=0.98]{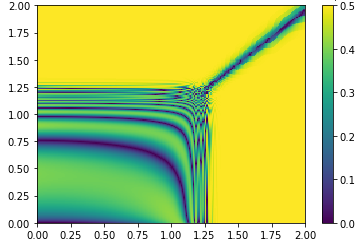}
		\label{fig:gini_money}
	}
	\subfigure[Amounts.]
	{
		\includegraphics[scale = 0.98]{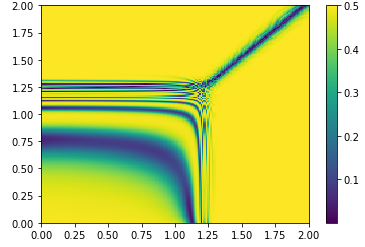}
		\label{fig:gini_amounts}
	}
	\caption{Gini coefficient for the two-player economy $\vec{a} = [[x,0.1], [15,y]]$ where the $x$ and $y$ values are represented on the $X$ and $Y$ axes. Initial amounts are $1$ and all initial bids are $0.5$.}
	\label{fig:gini_all}
\end{figure}


\begin{figure}[h!]
	\centering
	\subfigure[Budgets.]
	{
		\includegraphics[scale=0.98]{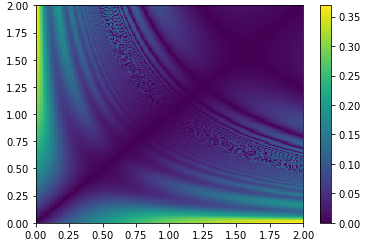}
		\label{fig:gini_money_cross}
	}
	\subfigure[Amounts.]
	{
		\includegraphics[scale = 0.98]{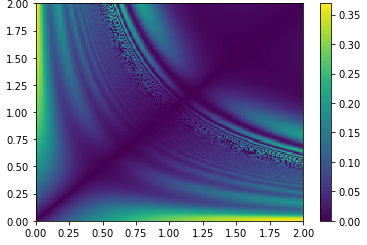}
		\label{fig:gini_amounts_cross}
	}
	\caption{Gini coefficient for the two-player economy $\vec{a} = [[1,x], [y,1]]$ where the $x$ and $y$ values are represented on the $X$ and $Y$ axes. The initial amounts are $1$, the initial bids $0.5$, and the coefficient is computed after 120 iterations.}
	\label{fig:gini_all_cross}
\end{figure}

\begin{figure}[h!]
	\centering
	\subfigure[Budgets.]
	{
		\includegraphics[scale=1]{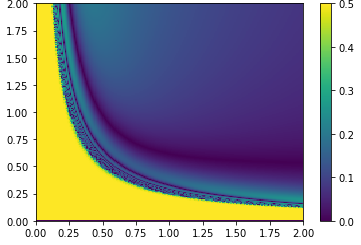}
		\label{fig:gini_money_asym}
	}
	\subfigure[Amounts.]
	{
		\includegraphics[scale = 1]{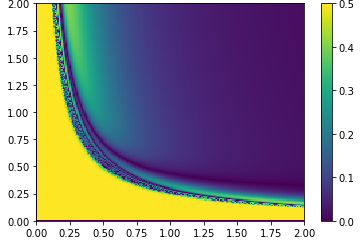}
		\label{fig:gini_amounts_asym}
	}
	\caption{Gini coefficient for the two-player economy $\vec{a} = [[0,x], [y,0.5]]$ where the $x$ and $y$ values are represented on the $X$ and $Y$ axes. The initial amounts are $1$, initial bids $0.5$.}
	\label{fig:gini_all_asym}
\end{figure}


\begin{figure}[h!]
	\centering
	\subfigure[Budgets.]
	{
		\includegraphics[scale=1]{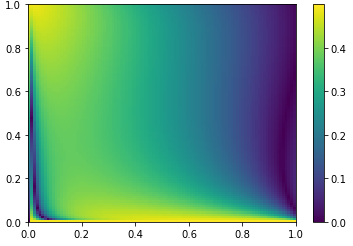}
		\label{fig:gini_money_bids_xxsmallselfloops}
	}
	\subfigure[Amounts.]
	{
		\includegraphics[scale = 1]{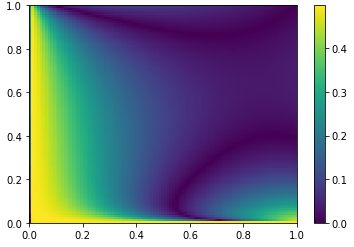}
		\label{fig:gini_amounts_bids_xxsmallselfloops}
	}
	\caption{Gini coefficient for the two-player economy $\vec{a} = [[0.25, 0.1], [15,0.25]]$ with initial bids $[[1-x,x], [y, 1-y]]$ and initial amounts $1$.}
	\label{fig:gini_bids_xxsmallselfloops}
\end{figure}

\begin{figure}[h!]
	\centering
	\subfigure[Budgets.]
	{
		\includegraphics[scale=1]{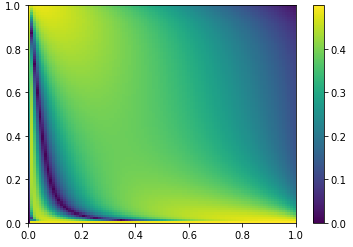}
		\label{fig:gini_money_bids_xsmallselfloops}
	}
	\subfigure[Amounts.]
	{
		\includegraphics[scale = 1]{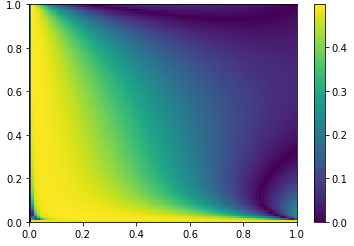}
		\label{fig:gini_amounts_bids_xsmallselfloops}
	}
	\caption{Gini coefficient for the two-player economy $\vec{a} = [[0.25, 0.1], [15,0.25]]$ with initial bids $[[1-x,x], [y, 1-y]]$ and initial amounts $1$.}
	\label{fig:gini_bids_xsmallselfloops}
\end{figure}

\begin{figure}[h!]
	\centering
	\subfigure[Budgets.]
	{
		\includegraphics[scale=1]{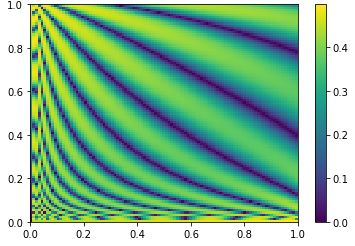}
		\label{fig:gini_money_bids}
	}
	\subfigure[Amounts.]
	{
		\includegraphics[scale = 1]{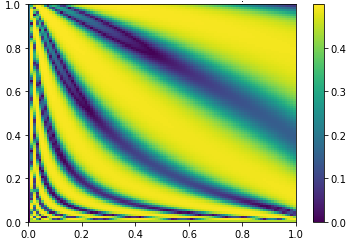}
		\label{fig:gini_amounts_bids}
	}
	\caption{Gini coefficient for the two-player economy $\vec{a} = [[1, 0.1], [15,1]]$ with initial bids $[[1-x,x], [y, 1-y]]$ and initial amounts $1$.}
	\label{fig:gini_bids}
\end{figure}

\begin{figure}[h!]
	\centering
	\subfigure[Budgets.]
	{
		\includegraphics[scale=1]{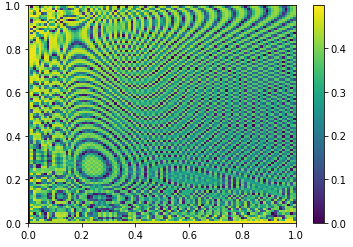}
		\label{fig:gini_money_bids_smallerselfloops}
	}
	\subfigure[Amounts.]
	{
		\includegraphics[scale = 1]{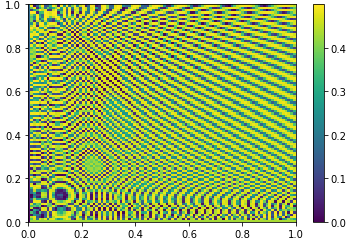}
		\label{fig:gini_amounts_bids_smallerselfloops}
	}
	\caption{Gini coefficient for the two-player economy $\vec{a} = [[1.21, 0.1], [15,1.21]]$ with initial bids $[[1-x,x], [y, 1-y]]$ and initial amounts $1$.}
	\label{fig:gini_bids_smallerselfloops}
\end{figure}

\begin{figure}[h!]
	\centering
	\subfigure[Budgets.]
	{
		\includegraphics[scale=1]{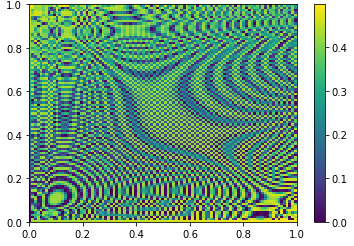}
		\label{fig:gini_money_bids_equalselfloops}
	}
	\subfigure[Amounts.]
	{
		\includegraphics[scale = 1]{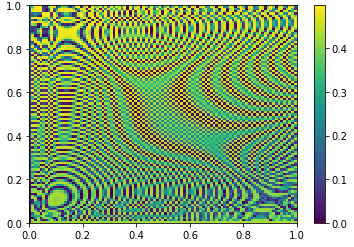}
		\label{fig:gini_amounts_bids_equalselfloops}
	}
	\caption{Gini coefficient for the two-player economy $\vec{a} = [[\sqrt{1.5}, 0.1], [15,\sqrt{1.5}]]$ with initial bids $[[1-x,x], [y, 1-y]]$ and initial amounts $1$.}
	\label{fig:gini_bids_equalselfloops}
\end{figure}

\begin{figure}[h!]
	\centering
	\subfigure[Budgets.]
	{
		\includegraphics[scale=1]{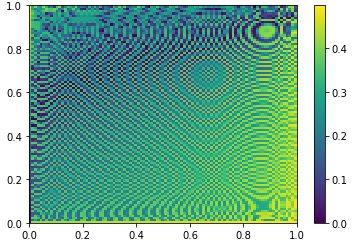}
		\label{fig:gini_money_bids_largeselfloops}
	}
	\subfigure[Amounts.]
	{
		\includegraphics[scale = 1]{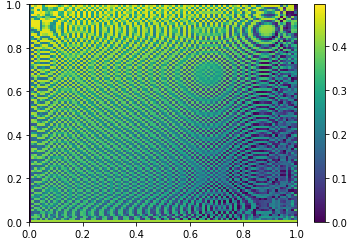}
		\label{fig:gini_amounts_bids_largeselfloops}
	}
	\caption{Gini coefficient for the economy $\vec{a} = [[1.24, 0.1], [15,1.24]]$ with initial bids $[[1-x,x], [y, 1-y]]$ and initial amounts $1$.}
	\label{fig:gini_bids_largeselfloops}
\end{figure}

\begin{figure}[h!]
	\centering
	\subfigure[Budgets.]
	{
		\includegraphics[scale=1]{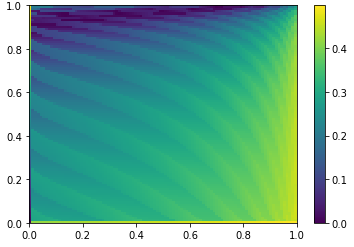}
		\label{fig:gini_money_bids_xlargeselfloops}
	}
	\subfigure[Amounts.]
	{
		\includegraphics[scale = 1]{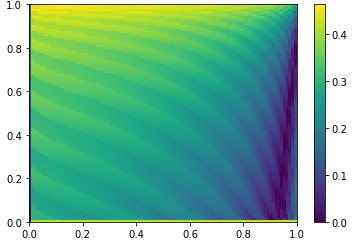}
		\label{fig:gini_amounts_bids_xlargeselfloops}
	}
	\caption{Gini coefficient for the two-player economy $\vec{a} = [[1.4, 0.1], [15,1.4]]$ with initial bids $[[1-x,x], [y, 1-y]]$ and initial amounts $1$.}
	\label{fig:gini_bids_xlargeselfloops}
\end{figure}

\begin{figure}[h!]
	\centering
	\subfigure[Budgets.]
	{
		\includegraphics[scale=1]{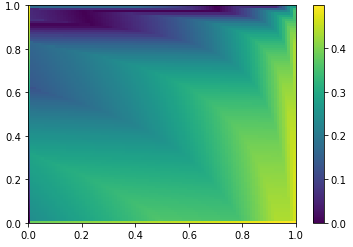}
		\label{fig:gini_money_bids_xxlargeselfloops}
	}
	\subfigure[Amounts.]
	{
		\includegraphics[scale = 1]{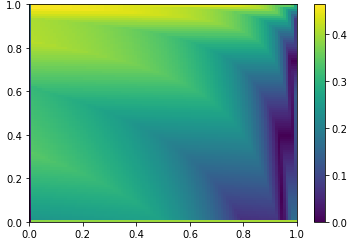}
		\label{fig:gini_amounts_bids_xxlargeselfloops}
	}
	\caption{Gini coefficient for the two-player economy $\vec{a} = [[2, 0.1], [15,2]]$ with initial bids $[[1-x,x], [y, 1-y]]$ and initial amounts $1$.}
	\label{fig:gini_bids_xxlargeselfloops}
\end{figure}

\end{document}